\documentclass{elsarticle}
\pdfoutput=1
\makeatletter
\def\ps@pprintTitle{%
  \let\@oddhead\@empty
  \let\@evenhead\@empty
  \def\@oddfoot{\reset@font\hfil\thepage\hfil}
  \let\@evenfoot\@oddfoot
}
\makeatother

\usepackage[T1]{fontenc}
\usepackage[latin1]{inputenc}
\usepackage[english]{babel}

\usepackage{amsmath,amsfonts,amssymb,amstext,amsthm} 
\usepackage{mathtools}

\usepackage[usenames,dvipsnames,table]{xcolor} 
\usepackage{tikz}
\usetikzlibrary{calc,positioning}
\usepackage{graphicx}

\usepackage{wrapfig}
\usepackage{booktabs} 
\usepackage{longtable}
\usepackage{float}

\usepackage{enumerate}
\usepackage{xspace}
\usepackage{ifthen}
\usepackage{fixltx2e}
\usepackage{url}
\usepackage{scrtime}

\usepackage[noindentafter]{titlesec}

\usepackage[ruled,longend,vlined]{algorithm2e}

\newcommand{\MSOii}{\ensuremath{\text{MSO}_2}\xspace}

\newcommand{\EMSOii}{\ensuremath{\text{EMSO}_2}\xspace}

\newcommand{\N}{\ensuremath{\mathbf{N}}\xspace} 
\newcommand{\R}{\ensuremath{\mathbf{R}}\xspace}	
\renewcommand{\cal}{\mathcal}

{\bfseries \upshape}{\itshape}
{\bfseries \upshape}{\itshape}

\newcommand{\name}[1]{\textsc{#1}}

\newcommand{\YES}{\textsc{yes}}

\newcommand{\mc}{\mathcal}

\newcommand{\G}{\mathcal{G}}

\newcommand{\PPP}{\mathfrak{P}}
\newcommand{\bigO}{\mathcal{O}}

\newcommand{\Problem}[1]{\textsc{#1}\xspace}

\newcommand{\diam}{\ensuremath{\mathop{diam}}} 
\newcommand{\minor}{\ensuremath{\preceq_{{\mathit{m}}}\!}} 
\newcommand{\grad}{\ensuremath{\nabla}}

\newcommand{\notminor}{\ensuremath{\npreceq_{{\mathit{m}}}\!}}

\newcommand{\bound}{\ensuremath{\mathop{bd}}}    
\newcommand{\Rep}{\mathcal{R}(t,d)}
\newcommand{\prot}{\ensuremath{\rho(t,d)}}   
\newcommand{\YYYY}{\ensuremath{Y_0 \uplus Y_1 \uplus \cdots \uplus Y_\ell}\xspace} 
\newcommand{\primedYYYY}{\ensuremath{Y_0' \uplus Y_1' \uplus \cdots \uplus Y_\ell'}\xspace} 
\renewcommand{\le}{\leqslant}
\renewcommand{\leq}{\leqslant}

\renewcommand{\geq}{\geqslant}

\newcommand{\widthm}[1]{\ensuremath{\mathop\mathbf{#1}}\xspace}
\newcommand{\tw}{\widthm{tw}}
\newcommand{\pw}{\widthm{pw}}
\newcommand{\td}{\widthm{td}}
\newcommand{\width}{\widthm{width}}

\newcommand{\probpar}{\xi} 
\newcommand{\nab}{\mathop{\triangledown}}

\newcommand{\FII}{FII\xspace} 

\newcommand{\clos}{\ensuremath{\text{clos}}}

\DeclareMathOperator{\degree}{deg}

\newtheorem{theorem}{Theorem}[section]
\newtheorem{lemma}[theorem]{Lemma}
\newtheorem{corollary}[theorem]{Corollary}

\newtheorem{proposition}[theorem]{Proposition}
\newtheorem*{claim}{Claim}

\theoremstyle{definition}
\newtheorem{redrule}[theorem]{Reduction Rule}
\newtheorem{definition}[theorem]{Definition}

\newcommand*{\ie}{i.e.\@\xspace}
\newcommand*{\cf}{cf.\@\xspace}
\newcommand*{\wrt}{w.r.t.\@\xspace}
\makeatletter
\newcommand*{\etc}{%
    \@ifnextchar{.}%
        {etc}%
        {etc.\@\xspace}%
}
\makeatother
\newenvironment{tightcenter}
 {\parskip=0pt\par\nopagebreak\centering}
 {\par\noindent\ignorespacesafterend}

\usepackage{ctable}
\newlength{\RoundedBoxWidth}
\newsavebox{\GrayRoundedBox}
\newenvironment{GrayBox}[1]%
   {\setlength{\RoundedBoxWidth}{\textwidth-4.5ex}
    \def\boxheading{#1}
    \begin{lrbox}{\GrayRoundedBox}
       \begin{minipage}{\RoundedBoxWidth}}%
   {   \end{minipage}
    \end{lrbox}
    \begin{tightcenter}
    \begin{tikzpicture}%
       \node(Text)[draw=black!20,fill=white,rounded corners,%
             inner sep=2ex,text width=\RoundedBoxWidth]%
             {\usebox{\GrayRoundedBox}};
        \coordinate(x) at (current bounding box.north west);
        \node [draw=white,rectangle,inner sep=3pt,anchor=north west,fill=white] 
        at ($(x)+(6pt,.75em)$) {\boxheading};
    \end{tikzpicture}
    \end{tightcenter}}             

\newenvironment{defproblem}[1]{\noindent\ignorespaces%
                                \FrameSep=6pt%
                                \parindent=0pt%
                \begin{GrayBox}{\textsc{#1}}%
                \begin{tabular*}{\textwidth}{@{\hspace{.1em}} >{\itshape} p{1.6cm} p{0.8\textwidth} @{}}%
            }{
                \end{tabular*}%
                \end{GrayBox}%
                \ignorespacesafterend%
             }    

\usepackage{framed}
\usepackage{footnote}
\renewenvironment{leftbar}[2]{%
  \def\FrameCommand{\vrule width #1 \hspace{#2}}\MakeFramed {\advance\hsize-\width \FrameRestore}}%
{\endMakeFramed}

\def\compxy(#1){\mathcal{U}_{#1}}

\def\Nesetril{Ne\v{s}et\v{r}il\xspace}

\title{Kernelization Using Structural Parameters on Sparse Graph Classes\tnoteref{t1}}

\tnotetext[t1]{Research funded by 
DFG-Project RO 927/12-1 ``Theoretical and Practical Aspects of Kernelization'',
the Czech Science Foundation under grant 14-03501S,
and the European Social Fund and the state budget of the Czech Republic
under project CZ.1.07/2.3.00/30.0009 (S.~Ordyniak).}

\author[muni]{Jakub Gajarsk\'{y}}
\author[muni]{Petr Hlin\v{e}n\'{y}}
\author[muni]{Jan Obdr\v{z}\'{a}lek}
\author[muni]{Sebastian Ordyniak}
\author[rwth]{Felix Reidl}
\author[rwth]{Peter Rossmanith}
\author[rwth]{Fernando S\'{a}nchez Villaamil}
\author[rwth]{Somnath Sikdar}
\address[muni]{\small Faculty of Informatics, Masaryk University,  
Brno, Czech Republic, 
\texttt{\{gajarsky,hlineny,obdrzalek,ordyniak\}@fi.muni.cz}.
}
\address[rwth]{\small Theoretical Computer Science, Department of Computer Science,
RWTH Aachen University, Aachen, Germany, 
\texttt{\{reidl,rossmani,fernando.sanchez,sikdar\}@cs.rwth-aachen.de}.
}
\date{}

\begin{document}

\begin{frontmatter}
    \begin{abstract}
    Meta-theorems for polynomial (linear) kernels 
have been the subject of intensive research in parameterized complexity.
Heretofore, meta-theorems for linear kernels exist
on graphs of bounded genus, $H$-minor-free graphs, 
and $H$-topological-minor-free graphs. To the best 
of our knowledge, no meta-theorems for polynomial kernels are known 
for any larger sparse graph classes; e.g.,
for classes of bounded expansion or for nowhere dense ones.
In this paper we prove such meta-theorems for the two latter cases.
More specifically, we show that graph problems that have finite integer index (\FII)
have linear kernels on graphs of bounded expansion when 
parameterized by the size of a modulator to constant-treedepth graphs. 
For nowhere dense graph classes, our result yields almost-linear kernels.
While our parameter may seem rather strong, we argue
that a linear kernelization result on graphs of bounded expansion 
with a weaker parameter (than treedepth modulator) would fail to 
include some of the problems covered by our framework.    
Moreover, we only require the problems to have \FII on graphs of 
constant treedepth. This allows us to prove linear kernels for problems such as 
\name{Longest Path/Cycle}, 
\name{Exact $s,t$-Path},
\name{Treewidth},
and \name{Pathwidth},
which do not have \FII on general graphs
(and the first two not even on bounded treewidth graphs).

    \end{abstract}
\end{frontmatter}

\section{Introduction}\label{sec:Introduction}
Data preprocessing has always been a part of algorithm design.
The last decade has seen steady progress in the area of \emph{kernelization},
an area which deals with the design of polynomial-time preprocessing algorithms. 
These algorithms compress an input instance of a parameterized problem into an equivalent 
output instance whose size is bounded by some function of the parameter.
Parameterized complexity theory guarantees the existence of such \emph{kernels}
for problems that are \emph{fixed-parameter tractable}. Some problems
admit stronger kernelization in the sense that the size of the output instance is bounded by 
a polynomial (or even linear) function of the parameter, the  
so-called \emph{polynomial (or linear) kernels}. \looseness-1

Of great interest are \emph{algorithmic meta-theorems}, results that 
focus on problem classes instead of single problems.  In the area of 
graph algorithms, such meta-theorems usually have the following form: 
all problems with a specific property admit, on a specific graph class, 
an algorithm of a specific type. We are specifically interested
in meta-theorems that concern kernelization, for which a solid groundwork
already exists. Before we delve into the history, we need to quickly
establish the keystone property that drives all these meta-theorems:
the notion of \emph{finite integer index} (\FII).

Roughly speaking, a graph problem has FII if there exists a
finite set~$\cal S$ of graphs such that every instance of the problem can be
represented by a member of~$\cal S$ alongside an integer ``offset''. 
This property is the basis of the
\emph{protrusion replacement rule} whereby protrusions (pieces of the input
graph satisfying certain requirements) are replaced by members of the
set~$\cal S$.  Finite integer index is an intrinsic
property of the problem itself and is not directly related to whether it can
be expressed in a certain logic. In particular, \MSOii expressibility does not
imply \FII (see~\cite{BFLPST09} for sufficiency conditions for a problem
expressible in counting MSO to have \FII). As an example of this phenomenon,
\name{Hamiltonian Path} has \FII on general graphs whereas \name{Longest Path}
does not, although both are \EMSOii-expressible. 

Now, the first steps towards a kernelization meta-theorem 
appeared in a paper by Guo and
Niedermeier who provided a prescription of how to design linear kernels on
planar graphs for graph problems which satisfy a certain distance
property~\cite{GN07a}.
Their work built on the seminal paper by Alber, Fellows, and Niedermeier who
showed that \name{Dominating Set} has a linear kernel on planar
graphs~\cite{AFN04}. This was followed by the first true meta-theorem in this
area by Bodlaender et al.~\cite{BFLPST09} who showed that graph problems that
have \FII and satisfy a
property called \emph{quasi-coverable}\footnote{This property was called
\emph{quasi-compactness} in earlier version of~\cite{BFLPST09}}, admit linear kernels on bounded
genus graphs. Shortly after~\cite{BFLPST09} was published, Fomin et
al.~\cite{FLST10} proved a meta-theorem for linear kernels on $H$-minor-free
graphs, a graph class that strictly contains graphs of bounded genus. A rough
statement of their main result states that any graph problem that has \FII, is
\emph{contraction bidimensional}, and satisfies a \emph{separation property} 
has a linear kernel on graphs that exclude a fixed graph as minor. This result was, in turn, generalized
in~\cite{KLPRRSS12} to $H$-topological-minor-free graphs, which strictly
contain $H$-minor-free graphs. Here, the problems are required to have \FII 
and to be \emph{treewidth-bounding}: A graph problem is
treewidth-bounding if \YES-instances have a vertex set of size
linear in the parameter, the deletion of which results in a graph of bounded treewidth.
Such  a vertex set is called a 
\emph{modulator to bounded treewidth}. Prototypical problems that satisfy this
condition are \name{Feedback Vertex Set} and \name{Treewidth $t$-Vertex
Deletion}\footnote{For problem definitions, see Appendix.}, when parameterized
by the solution size.

We see that while these meta-theorems (viewed in chronological order) steadily covered 
larger graph classes, the set of problems captured in their framework diminished 
as the other precondition(s) became stricter. Surprisingly, this is not due
to said preconditions: It turns out that they can be
expressed in a unified manner and are therefore equally restrictive.
The combined properties of
bidimensionality and separability (used to prove
the result on $H$-minor-free graphs) imply that the problem is treewidth-%
bounding (\cf Lemma~3.2 and~3.3 in~\cite{FLST10}). Quasi-coverability
on bounded genus graphs implies the same (\cf Lemma~6.4 in~\cite{BFLPST09}). 
This demonstrates that all three previous meta-theorems on linear kernels
implicitly or explicitly used treewidth-boundedness. Hence the diminishing
set of problems can be blamed on the increasingly weaker
interaction of the graph classes with the problem parameters, not the (only
apparently) stricter precondition on the problems.

This insight motivates a different view on previous meta-theorems: 
problems that have \FII admit linear kernels if parameterized by
a \emph{treewidth modulator}. This view replaces the natural problem parameter---whose
structural impact diminishes in larger sparse graph classes---by an explicit
\emph{structural} parameter which retains the crucial interaction between
parameter and graph class. It also gives us, as we will see, the freedom
to adapt the parameterization to our needs. 

The next well-established level in the sparse-graph hierarchy \cite{NOdM12}
is formed by the classes of \emph{bounded expansion}. The notion was introduced by 
\Nesetril and Ossona de Mendez~\cite{NOdM08} and subsumes graph classes
excluding a fixed graph as a topological minor. It turns out that for
these classes the serviceable parametrization by a treewidth modulator
cannot work if we aim for linear kernels:
Any graph class~$\cal G$ can be transformed into a
class~$\tilde{\cal G}$ of bounded expansion by replacing every graph $G \in
\cal G$ with~$\tilde{G}$, obtained in turn by replacing each edge of~$G$ by a path on
$|V(G)|$ vertices. For problems like \name{Treewidth $t$-Vertex Deletion} and,
in particular, \name{Feedback Vertex Set} this operation neither changes the
instance membership nor does it increase the parameter. As both the problems do
not admit kernels of size $O(k^{2-\epsilon})$ unless coNP $\subseteq$
NP/poly, by a result of Dell and Melkebeek~\cite{DM10}, 
a linear kernelization result on bounded-expansion classes of graphs 
and under the treewidth-modulator parameterization 
would have to exclude both these natural problems. 

In this work, we identify a structural parameter that indeed does allow linear kernels for
all problems that have \FII on graph classes of bounded expansion---%
the size of a tree\emph{depth} modulator. This
parameter not only increases under edge subdivisions (a necessary prerequisite
as we now know), but it also provides exactly the structure that seems
necessary to obtain such a result. To put this parameterization
into context, let us recap some previous work on structural parameters.
Even outside the realm of sparse graphs, they have been
used to zero in on those aspects of problems that make them intractable---a 
development that certainly fits the overall agenda of parameterized complexity.
This research of alternative parameterizations has given rise to what is 
called the \emph{parameterized ecology}~\cite{FJR13}. 

Already the perhaps strongest structural parameter for graph-related 
problems---the vertex cover number---makes up an interesting nieche of said
ecology, as we summarize now.
Many problems that are W-hard or otherwise difficult to parameterize such as
\name{Longest Path}~\cite{BJK11}, \name{Cutwidth}~\cite{CLPPS11},
\name{Bandwidth}, \name{Imbalance}, \name{Distortion}~\cite{FLMRS08},
\name{List Coloring},
\name{Precoloring Extension}, \name{Equitable Coloring}, \name{L($p$,$1$)-Labeling},
and \name{Channel Assignment}~\cite{FGK09} are (easily) fpt when parameterized by
the vertex cover number. 
Some generalizations of vertex cover have also
been successfully used as a parameter, e.g.,~\cite{DK12,Gan11}. 
Even problems that do admit kernels in general or are fpt can benefit
from such a strong structural parameter---for example,
\name{Odd Cycle Transversal} (which admits a randomized and highly technical kernel), 
\name{Chordal Deletion} (which is fpt but does not admit a polynomial kernel), and
\name{$\cal F$-Minor-Free Deletion}~\cite{FJP12}. 
On the other hand, some problems do not
admit polynomial kernelization even with this strong additional parametrization:
\name{Dominating Set}, for example, has no polynomial kernel when
parameterized by the solution size \emph{and} the vertex cover number~\cite{DLS09}.

In light of previous work on structural parameters 
and the fact that a modulator to bounded
treedepth is a significantly weaker parameter than the vertex cover number (which is
the special case of a modulator to treedepth one), we conclude that
treedepth modulator is a well-motived choice in our case.

\subsubsection*{Our contribution} 
We show that, assuming \FII, a parameterization by
the size of a modulator to bounded treedepth allows for linear
kernels in linear time on graph classes of bounded expansion. The same parameter
yields almost-linear kernels on nowhere dense graph classes, which strictly contain 
those of bounded expansion. In particular, nowhere dense
classes are the largest collections of graphs that may still be called sparse~\cite{NOdM12}. 
In these results we do not require a treedepth modulator to be supplied as
part of the input, as we show that it can be approximated to within a constant factor. 

Furthermore, we only need \FII to hold on graphs of bounded treedepth, thus including
problems which do not have \FII in general. Some problems that are included
because of this relaxation are  \name{Longest Path/Cycle}, \name{Pathwidth} and
\name{Treewidth}, none of which have polynomial kernels with respect 
to their standard parameters, even on sparse graphs, since they admit simple
AND/OR-Compositions~\cite{BDFH09}. Problems covered by our framework include also
\name{Hamiltonian Path/Cycle}, several variants of \name{Dominating Set}, 
\name{(Connected) Vertex Cover}, \name{Chordal Vertex Deletion}, 
\name{Feedback Vertex Set}, \name{Induced Matching}, \name{Branchwidth}
and \name{Odd Cycle Transversal}.
In particular, we cover all problems included in earlier frameworks~\cite{BFLPST09,FLST10,KLPRRSS12}.
We emphasize, however, that this paper does not
subsume the former results due to our stricter parameter.

\subsubsection*{Organisation}
Our notation and the main definitions pertaining to graph classes can
all be found in Section~\ref{sec:Preliminaries}.
Section~\ref{sec:Protrusions} deals with the notion of finite integer index
and the protrusion machinery. In Section~\ref{sec:Kernels}, we prove our meta-
theorems for graph classes of bounded expansion and for nowhere dense ones.
Section~\ref{sec:fii_pw_tw} is devoted to the proof that the problems
\name{Treewidth} and \name{Pathwidth} have \FII in appropriate graph classes.
We conclude in Section~\ref{sec:Conclusion} with some open problems.
In the appendix, we define (some of) the graph-algorithmic problems that
we deal with in this paper.

\section{Preliminaries}\label{sec:Preliminaries}
\noindent We use standard graph-theoretic notation (see~\cite{Die10} for any undefined
terminology). All our graphs are finite and simple. Given a graph $G$, we
use $V(G)$ and $E(G)$ to denote its vertex and edge sets. For convenience we
assume that $V(G)$ is a totally ordered set, and use $uv$ instead of $\{u,v\}$ to denote the edges of
$G$. For~$X \subseteq V(G)$, we let~$G[X]$ denote the subgraph of $G$
induced by $X$, and we define $G-X:=G[V(G) \setminus X]$.  Since we will
mainly be concerned with sparse graphs in this paper, we let $|G|$ denote
the number of vertices in the graph $G$. The distance $d_G(v,w)$ of two
vertices $v,w \in V(G)$ is the length (number of edges) of a shortest
$v,w$-path in $G$ and $\infty$ if $v$ and $w$ lie in different connected
components of $G$. The diameter $\diam(G)$ of a graph is the length of
the longest shortest path between all pairs of vertices in $G$.
A complete subgraph of $G$ is called a {\em clique} and
we denote by $\omega(G)$ the largest size of a clique of
$G$.

The concept of neighborhood is used heavily throughout the paper.
The neighborhood of a vertex $v \in V(G)$ is the set $N^G(v)=\{w\in V(G)
\mid vw\in E(G)\}$, the degree of $v$ is $\degree^G(v)=|N^G(v)|$, and the closed neighborhood of~$v$ is defined as $N^G[v] :=
N^G(v) \cup \{v\}$. We extend this naturally to sets of vertices and
subgraphs: For $S\subseteq V(G)$ we denote $N^G(S)$ the set of vertices in
$V(G)\setminus S$ that have at least one neighbor in $S$, and for a subgraph
$H$ of $G$ we put $N^G(H)=N^G(V(H))$. Finally if $X$ is a subset of vertices
disjoint from $S$, then $N_X^G(S)$ is the set $N^G(S)\cap X$ (and similarly
for $N_X^G(H)$). Given a graph~$G$ and a set~$W \subseteq V(G)$, we also
define~$\partial_G(W)$ as the set of vertices in~$W$ that have a neighbor
in~$V \setminus W$. Note that $N^G(W) = \partial_G(V(G) \setminus W)$. A
graph $G$ is $d$-degenerate if every subgraph of $G' \subseteq G$ contains a vertex 
$v \in V(G')$ with $deg^G(v)\leq d$. The degeneracy of $G$ is the smallest 
$d$ such that $G$ is $d$-degenerate.

In the rest of the paper we often drop the index $G$ from all the notation
if it is clear which graph is being referred to.

\subsection{Minors and shallow minors}
We start by defining the notion of edge contraction. Given an edge~$e = uv$ of a graph~$G$, we let~$G/e$ denote the graph obtained
from~$G$ by \emph{contracting} the edge~$e$, which amounts to deleting the
endpoints of~$e$, introducing a new vertex~$w_{uv}$, and making it adjacent to
all vertices in $(N^G(u) \cup N^G(v)) \setminus \{u,v\} $. By contracting $e = uv$ 
to the vertex~$w$, we mean that the vertex~$w_{uv}$ is renamed
as~$w$. \emph{Subdividing} an edge is, in a sense, an opposite operation to contraction.
A graph $G$ is called a \emph{$\leq\!k$-subdivision} of a graph $H$ if
(some) edges of $H$ are replaced by paths of length at most~$k+1$.

A \emph{minor} of~$G$ is a graph obtained from a subgraph of~$G$ by contracting 
zero or more edges.
In a more general view, if $H$ is isomorphic to a minor of~$G$, then we call
$H$ a minor of $G$ as well, and we write~$H \minor G$.
A graph $G$ is {\em $H$-minor-free} if $H \notminor G$. 

We next introduce the notion of a shallow minor.
\begin{definition}[Shallow minor~\cite{NOdM12}] 
For an integer~$d$, a graph~$H$ is a \emph{shallow minor at depth~$d$} of~$G$ if there
exists a set of disjoint subsets $V_1, \ldots, V_p$ of~$V(G)$ such that
\begin{enumerate}
\item each graph $G[V_i]$ has radius at most~$d$, meaning that there exists~$v_i \in V_i$
(a \emph{center}) such that every vertex in~$V_i$ is within distance at most~$d$ in~$G[V_i]$;
\item there is a bijection $\psi \colon V(H) \rightarrow \{V_1, \ldots, V_p\}$ such that for $u,v \in V(H)$,
$uv \in E(H)$ only if there is an edge in~$G$ with an endpoint each in $\psi(u)$ and $\psi(v)$.
\end{enumerate}
\end{definition}
\noindent Note that if $u,v \in V(H)$, $\psi(u) = V_i$, and $\psi(v) = V_j$ then $d_G(v_i,v_j) \leq (2d+1) \cdot d_H(u,v)$. 
The class of shallow minors of~$G$ at depth~$d$ is denoted by $G \nab d$. This notation
is extended to graph classes~$\cal G$ as well: $\cal G \nab d = \bigcup_{G \in \cal G} G \nab d$. 

\subsection{Parameterized problems, kernels and treewidth}
In this paper we deal with parameterized problems where the value of the parameter is not explicitly 
specified in the input instance. This situation is slightly different from the usual case where the 
parameter is supplied with the input and a parameterized problem is defined as sets of tuples~$(x,k)$
as in~\cite{DF99}. As such, we find it convenient to adopt the definition of Flum and Grohe~\cite{FG06}
and we feel that this is the approach one might have to choose when dealing with generalized parameters
as is done in this paper. 

Let $\Sigma$ be a finite alphabet. A parameterization of $\Sigma^{*}$ is a mapping $\kappa \colon \Sigma^{*} \rightarrow \N_0$ 
that is polynomial time computable. A parameterized problem~$\Pi$ is a pair $(Q, \kappa)$ consisting of a set  
$Q \subseteq \Sigma^{*}$ of strings over~$\Sigma$ and a parameterization $\kappa$ over~$\Sigma^{*}$. 
A parameterized problem~$\Pi$ is \emph{fixed-parameter tractable} if there
exist an algorithm~$\mathcal{A}$, a computable function~$f \colon \N \rightarrow \N$ and a polynomial~$p$ 
such that for all $x \in \Sigma^{*}$, $\mathcal{A}$ decides~$x$ in time $f(\kappa(x)) \cdot p(|x|)$.

\begin{definition}[Graph problem]
    A \emph{graph problem}~$\Pi$ is a set of pairs $(G,\probpar)$, where~$G$ is a graph and 
    $\probpar \in \N_0$, such that for all graphs $G_1, G_2$ and all $\probpar \in \N_0$,
    $G_1 \cong G_2$ implies that $(G_1,\probpar) \in \Pi$  iff $(G_2,\probpar) \in \Pi$.
    For a graph class $\cal G$, we define $\Pi_{\cal G}$ as the set of pairs $(G, \probpar) \in \Pi$ 
    such that $G \in \cal G$.
\end{definition}

\begin{definition}[Kernelization]
    A \emph{kernelization} of a parameterized problem $(Q,\kappa)$ over the alphabet~$\Sigma$ 
   is a polynomial-time computable function $A \colon \Sigma^{*} \rightarrow \Sigma^{*}$ such that 
  for all $x \in \Sigma^{*}$, we have
    \begin{enumerate}
        \item $x \in Q$ if and only if $A(x) \in Q$,
        \item $|A(x)| \leq g(\kappa(x))$,
    \end{enumerate}
    where~$g$ is some computable function. The function $g$ is called
    the \emph{size} of the kernel. If $g(\kappa(x)) = \kappa(x)^{\bigO(1)}$ or $g(\kappa(x)) = \bigO(\kappa(x))$,
    we say that $\Pi$ admits a \emph{polynomial kernel} and a \emph{linear kernel}, respectively.
\end{definition}

\begin{definition}[Treewidth]
A \emph{tree decomposition}~$\mathcal{T}$ of an (undirected) graph $G=(V,E)$ is a pair 
$(T,\chi)$, where $T$ is a tree and $\chi$ is a function that assigns each 
tree node $t$ a set $\chi(t) \subseteq V$ of vertices such that the following 
conditions hold: 
\begin{itemize}
\item[(P1)] For every vertex~$u \in V$, there is a tree node~$t$ such that 
	$u \in \chi(t)$.
\item[(P2)] For every edge $\{u,v\} \in E(G)$ there is a tree node
  $t$ such that $u,v\in \chi(t)$.
\item[(P3)] For every vertex $v \in V(G)$,
  the set of tree nodes $t$ with $v\in \chi(t)$ forms a subtree of~$T$.
\end{itemize}
The sets $\chi(t)$ are called \emph{bags} of the decomposition~$\mathcal{T}$ and $\chi(t)$ 
is the bag associated with the tree node~$t$. 
The \emph{width} of a tree decomposition 
$(T,\chi)$ is the size of a largest bag minus~$1$.  A tree decomposition of
minimum width is called \emph{optimal}.  The \emph{treewidth} of a graph $G$,
denoted by $\textup{tw}(G)$, is the
width of an optimal tree decomposition of~$G$. 
\end{definition}

Let $\mathcal{T}=(T,\chi)$ be a tree decomposition of a graph $G$ and let
$G'$ be an induced subgraph of $G$. The \emph{projection} of $\mathcal{T}$
onto $G'$, denoted by $\mathcal{T}|G'$,
is the pair $(T,\chi')$ where $\chi'(t)=\chi(t) \cap V(G')$ for every $t \in
V(T)$. It is well-known that $\mathcal{T}|G'$ is a tree decomposition of
$G'$.

A \emph{path decomposition} of a graph $G$ is a tree decomposition
$(T,\chi)$ such that $T$ is a path instead of a tree. 
All notions and 
definitions introduced for tree decompositions above apply in the same way
for path decompositions. 
The
\emph{pathwidth} of $G$, denoted by $\textup{pw}(G)$, is the width of an
optimal path decomposition of $G$.

\subsection{Grad and graph classes of bounded expansion}

Let us recall the main definitions pertaining to the notion of graphs of
bounded expansion. We follow the recent book by \Nesetril and Ossona de Mendez~\cite{NOdM12}.
\begin{definition}[Greatest reduced average density (grad)~\cite{NOdM08,NMW12}]
	Let $\cal G$ be a graph class. Then the \emph{greatest reduced average density} of $\cal G$
	with \emph{rank $d$} is defined as 
	\[
		\grad_d(\cal G) = \sup_{H \in \cal G \nab d} \frac{|E(H)|}{|V(H)|}.
	\]
\end{definition}

\noindent This notation is also used for graphs via the convention that $\grad_d(G) := \grad_d(\{G\})$.
In particular, note that $G \nab 0$ denotes the set of subgraphs of $G$ and hence $2\grad_0(G)$ is 
the maximum average degree of all subgraphs of~$G$. The \emph{degeneracy} of~$G$ is, therefore, exactly~$2\grad_0(G)$. 
\begin{definition}[Bounded expansion~\cite{NOdM08}]
	A graph class $\cal G$ has \emph{bounded expansion} if there exists a
	function $f \colon \N \rightarrow \R$ (called the \emph{expansion function}) such 
	that for all $d \in \N$, $\grad_d(\cal G) \leq f(d)$.
\end{definition}

\noindent If $\cal G$ is a graph class of bounded expansion with expansion function~$f$, we say that $\cal G$ 
has \emph{expansion bounded by $f$}. An important relation we make use of later is: $\grad_d(G) = \grad_0(G \nab d)$, 
\ie the grad of $G$ with rank $d$ is precisely one half the maximum average degree of subgraphs of its 
depth~$d$ shallow minors.

Another important notion that we make use of extensively is that of treedepth. 
In this context, a \emph{rooted forest} is a disjoint union of rooted trees. 
For a vertex~$x$ in a tree~$T$ 
of a rooted forest, the \emph{height} (or {\em depth})
of~$x$ in the forest is the number of vertices in the path from 
the root of~$T$ to~$x$. The \emph{height of a rooted forest} is the maximum height of a vertex of the forest. 
\begin{definition}[Treedepth]
  Let the \emph{closure} of a rooted forest~$\cal F$ is the graph
  $\clos(\cal F)=(V_c,E_c)$ with the vertex set 
  $V_c=\bigcup_{T \in \cal F} V(T)$ and the edge set
  $E_c=\{xy \colon \text{$x$ is an ancestor of $y$ in some $T\in\cal F$}\}$.
  A \emph{treedepth decomposition}
  of a graph $G$ is a rooted forest $\cal F$ such that $G \subseteq \clos(\cal F)$.
  The \emph{treedepth} $\td(G)$ of a graph~$G$ is the minimum height of
  any treedepth decomposition of $G$. 
\end{definition}

\begin{proposition}[\hspace{-.02em}\cite{NOdM12}]\label{pro:compute-td}
  Given a graph $G$ with $n$ nodes and a constant $w$, it is possible to
  decide whether $G$ has treedepth at most $w$, and if so, to compute an
  optimal treedepth decomposition of $G$ in time $\bigO(n)$.
\end{proposition}

\noindent
We list some well-known facts about graphs of bounded treedepth. 
Omitted proofs can be found in~\cite{NOdM12}.
\begin{enumerate}
\item If a graph has no path with more than~$d$ vertices, then its treedepth is at most~$d$. 
\item If~$\td(G) \leq d$, then~$G$ has no paths with~$2^d$ vertices and, in particular, any DFS-tree 
of~$G$ has depth at most~$2^d - 1$.
\item If~$\td(G) \leq d$, then~$G$ is $d$-degenerate and hence has at most~$d \cdot |V(G)|$ edges.
\item If~$\td(G) \leq d$, then $\tw(G)\leq \pw(G) \leq d-1$.  
\end{enumerate} 
A useful way of thinking about graphs of bounded treedepth is that they are
(sparse) graphs with no long paths. 

For a graph~$G$ and an integer~$d$, a \emph{modulator to treedepth~$d$} of~$G$ 
is a set of vertices $M \subseteq V(G)$ such that $\td(G - M) \leq d$. 
The size of a modulator is the cardinality of the set $M$.

Finally, we need the following well-known result on degenerate graphs.

\begin{proposition}[\hspace{-0.02em}\cite{Woo07}]\label{prop:DegenerateNumCliques}
	Every $d$-degenerate graph $G$ with $n \geq d$ vertices has at most $2^d(n-d+1)$ cliques.
\end{proposition}

\section{The Protrusion Machinery}\label{sec:Protrusions}
In this section, we recapitulate the main ideas of the protrusion machinery developed
in~\cite{BFLPST09,FLST10}.  

\begin{definition}[$r$-protrusion~\cite{BFLPST09}]
    Given a graph~$G$, a set~$W \subseteq V(G)$ is a \emph{$r$-protrusion}
    of~$G$ if $|\partial_G (W)| \leq r$ and $\tw(G[W]) \le r-1$.%
    \footnote{We want the bags in a tree-decomposition of $G[W]$ to be of
      size at most~$r$.} We call~$\partial_G (W)$ the \emph{boundary} and $|W|$ the \emph{size}
    of the protrusion $W$. 
\end{definition}
\noindent
Thus an $r$-protrusion in a graph can be seen as a subgraph that is separated from the rest 
of the graph by a small boundary and, in addition, has small treewidth. See Figure~\ref{fig:Protrusion}.

\begin{figure}
\centering
  \includegraphics[width=10cm]{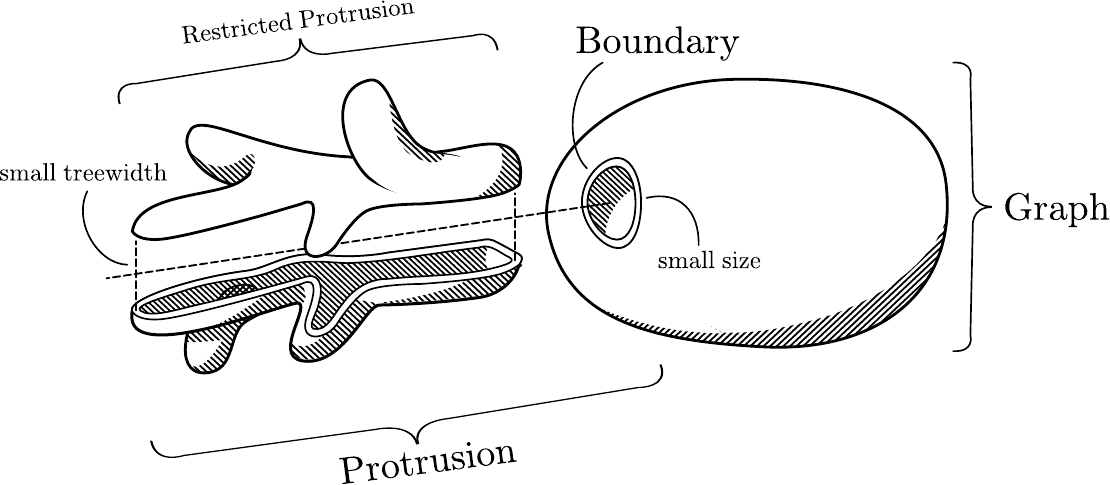}
  \caption{The anatomy of a protrusion.}\label{fig:Protrusion}
\end{figure}  

A \emph{$t$-boundaried graph} is a pair $(G, \bound(G))$, where 
$G$ is a graph and~$\bound(G) \subseteq V(G)$ is a set of~$t=|\bound(G)|$ 
vertices with distinct {\em labels} from the set $\{1, \ldots, t\}$. The graph~$G$ 
is called the \emph{underlying unlabeled graph} and $\bound(G)$ is called the
\emph{boundary}.\footnote{Usually denoted by $\partial(G)$, but this collides
with our usage of $\partial$.}  Given a
graph class $\cal G$, we let~$\cal G_{t}$ denote the class of $t$-boundaried
graphs $(G,\bound(G))$ where~$G \in \cal G$. 

For $t$-boundaried graphs $(H,\bound(H))$ 
and $(G,\bound(G))$, we say that $(H,\bound(H))$ is a subgraph of 
$(G,\bound(G))$ if $H \subseteq G$ and $\bound(H) = \bound(G)$. We say that 
$(H,\bound(H))$ is an induced subgraph of $(G,\bound(G))$ if for some $X \subseteq V(G)$, 
$H = G[X]$ and $\bound(H) = \bound(G)$. We say that the boundaries
of two $t$-boundaried graphs $(G, \bound(G))$ and 
$(H, \bound(H))$ are \emph{identical} if the function mapping 
each vertex of $\bound(G)$ to that vertex of $\bound(H)$ with the same label 
is an isomorphism between $G[\bound(G)]$ and $H[\bound(H)]$. 
Note that in the case of $(H, \bound(H))$ being an induced subgraph of 
$(G,\bound(G))$, the boundaries are identical by definition.
In the following, we will denote a $t$-boundaried graph $(G, \bound(G))$
shortly by~$\widetilde{G}$ to avoid cumbersome notation.

\begin{definition}[Gluing and ungluing]
    For $t$-boundaried graphs $\widetilde{G}_1$ and $\widetilde{G}_2$,
    we let~$\widetilde{G}_1 \oplus \widetilde{G}_2$ denote the
    graph obtained by taking the disjoint union of~$G_1$ and~$G_2$ and identifying
    each vertex in~$\bound(G_1)$ with the vertex in~$\bound(G_2)$ with the same label.
    The resulting order of vertices is an arbitrary extension of the
    orderings on $V(G_1)$ and $V(G_2)\setminus V(G_1)$.
    This operation is called \emph{gluing}.

    Let~$H \subseteq G$ and let~$B$ be a labeled vertex set
    consisting of~$\bound(H)$ with unique labels $\{1,\dots,t\}$.
    The operation of \emph{ungluing} $H$ from~$G$ creates the $t$-boundaried graph
    $G \ominus_{B} H := \big(G -(V(H) \setminus B), B \big)$.
\end{definition}
\noindent
The gluing operation entails taking the union of edges both of whose 
endpoints are in the boundary, with implicit deletion of multiple edges to keep the graph simple. 
The ungluing operation preserves the boundary (both the vertices and the
edges). For the sake of clarity, we sometimes annotate the $\oplus$ operator
with the boundary as well.

Note that an $r$-protrusion $W$ of a graph $G$ implicitly defines
a $t$-boundaried graph $\widetilde{G}[W]:=\big(G[W],\partial_G(W)\big)$, $\,t=|\partial_G(W)|\leq r$,
where the boundary vertices are assigned labels from $\{1, \ldots, t\}$
according to their order in $G$.
Hence we can rigorously deal with protrusions in $G$ as with $t$-boundaried
subgraphs of~$G$ as, e.g., in the following definition.
\begin{definition}[Replacement]
    Let~$W$ be an $r$-protrusion of a graph $G$ defining the $t$-boundaried
    graph $\widetilde{G}[W]$, and let $B$ be the labeled set of the boundary $\partial_G(W)$.
    For a $t$-boundaried graph~$\widetilde{H}$, 
    \emph{replacing}~$\widetilde{G}[W]$ by~$\widetilde{H}$ in $G$ is defined as the operation
    $
        (G \ominus_B G[W]) \oplus_B \widetilde{H}
    $.
\end{definition}

\noindent
The following definition concerns the centerpiece of our framework. 

\begin{definition}[Finite integer index; FII] \label{def:finiteii}
    Let~$\Pi$ be a graph problem and let $\widetilde{G}_1 = (G_1, \bound(G_1))$, 
    $\widetilde{G}_2 = (G_2, \bound(G_2))$ be two $t$-boundaried graphs.
    We say that $\widetilde{G}_1 \equiv_{\Pi, t} \widetilde{G}_2$ 
    if there exists an integer constant 
    $\Delta_{\Pi, t}(\widetilde{G}_1,\widetilde{G}_2)$ such that for all $t$-boundaried graphs 
    $\widetilde{H} = (H, \bound(H))$ and for all~$\probpar \in \N$:
    $$
        \big(\widetilde{G}_1 \oplus \widetilde{H}, \probpar\big) \in \Pi ~\text{iff}~
          \big(\widetilde{G}_2 \oplus \widetilde{H}, \probpar + \Delta_{\Pi, t}(\widetilde{G}_1,\widetilde{G}_2)\big) \in \Pi.
    $$
    We say that~$\Pi$ has \emph{finite integer index in the class 
    $\cal F$} if, for every~$t \in \N$, the relation $\equiv_{\Pi, t}$ has finite index
    if restricted to $\cal F$.
\end{definition}
\noindent
Note that the constant $\Delta_{\Pi, t}(\widetilde{G}_1,\widetilde{G}_2)$
depends on  $\Pi$,  $t$, and the \emph{ordered pair}  $(\widetilde{G}_1,
\widetilde{G}_2)$ so that  $\Delta_{\Pi, t}(\widetilde{G}_1,\widetilde{G}_2) =
-\Delta_{\Pi, t}(\widetilde{G}_2,\widetilde{G}_1)$. On most occasions, the
problem $\Pi$ and the class $\cal F$ will be clear from the context and in
such situations, we use $\equiv_{t}$ and $\Delta_{t}$ instead of $\equiv_{\Pi,
t}$ and  $\Delta_{\Pi, t}$, respectively.

If a graph problem has finite integer index then its instances can be
reduced by ``replacing protrusions''. The technique of replacing protrusions
hinges on the fact that each protrusion of ``large'' size can be replaced by a
``small'' gadget from the same equivalence class as the protrusion, which
consequently behaves similarly \wrt the problem at hand. If~$\widetilde{G}_1$ is replaced
by a gadget~$\widetilde{G}_2$, then~$\probpar$ changes 
by~$\Delta_{\Pi,t}(\widetilde{G}_1,\widetilde{G}_2)$. 
Many problems have finite integer index in general graphs including \name{Vertex Cover},
\name{Independent Set}, \name{Feedback Vertex Set}, \name{Dominating Set}, 
\name{Connected Dominating Set}, \name{Edge Dominating Set}. For a more complete list 
see~\cite{BFLPST09,FLST10}. Some problems that do not have finite integer index in general 
graphs are \name{Connected Feedback Vertex Set}, \name{Longest Path} and \name{Longest Cycle}.

For a graph class $\cal F$, let $\cal F_t$ denote the class of all $t$-boundaried 
graphs made of the members of $\cal F$.
The next lemma shows that if we assume that a graph problem $\Pi$ 
has FII in a graph class $\cal F$, then we can choose finitely many representatives
for the equivalence classes of $\equiv_{\Pi, t}$ from a (possibly different) 
graph class $\cal G$ under certain circumstances. 

\begin{lemma}\label{lemma:ProtInducedSubgraph}
    Let $\cal F$ be a graph class and $\Pi$ a graph problem such that $\Pi$ 
    has FII in $\cal F$.
    Let $\cal G$ be a class of graphs with vertex labels from
    $\{1,\dots,t\}$, and $\preceq$ be a relation on $\cal G$ such that
    $\cal G$ is well quasi-ordered by~$\preceq$.
    Then, for each $t \in \N$, there exists a finite set 
    $\cal R(t, \cal F, \cal G, \preceq) \subseteq \cal F_t \cap \cal G$
    with the following property.
    For every $\widetilde{G} = (G, \bound(G)) \in \cal F_t \cap \cal G$ 
    there exists $\widetilde{G}_0 = (G_0, \bound(G_0)) \in \cal R(t, \cal F, \cal G, \preceq)$ 
    such that it holds; $\widetilde{G} \equiv_{\Pi, t} \widetilde{G}_0$,\, 
    $\bound(G)\!$ and $\bound(G_0)$ are identical, and $G_0 \preceq G$.
\end{lemma}
\begin{proof}
    Let $\cal E_1, \ldots, \cal E_q$ be the equivalence classes of the relation 
    $\equiv_{\Pi,t}$ on $\cal F_t$, where~$q$ is some constant. 
    For each equivalence class $\cal E_i$, define $\cal E_{i}' =\cal E_i \cap \cal G$. 
    Next, partition $\cal E_{i}'$ into at most $2^{t^2} \cdot t!$ 
    sets $\cal E_{i,j}'$ such that all graphs in $\cal E_{i,j}'$ have identical boundaries.  
    Since $\cal G$ is well quasi-ordered by~$\preceq$, there is a
    \emph{finite} set $\cal G_{i,j}\subseteq \cal E_{i,j}'$ of
    the $\preceq$-minimal elements, for every $i,j$ as above.
    In other words, for all $\widetilde{G} \in \cal E_{i,j}'$ there exist
    $\widetilde{G}_0 \in \cal G_{i,j}$ satisfying the three properties stated in the lemma. 
    Consequently, $\bigcup_{j} \cal G_{i,j}$ can be chosen as the representatives for
    each $\cal E_i$.
    Altogether, define $\cal R(t, \cal F, \cal G, \preceq) = \bigcup_{i,j} \cal G_{i,j}$.
    Since $\cal R(t, \cal F, \cal G, \preceq)$ is the finite union of finite sets, it is finite.
\end{proof}
 
\noindent
Let us explain how we use Lemma~\ref{lemma:ProtInducedSubgraph}. The graph
problems~$\Pi$ that we consider in this paper usually have FII  on the class
of general graphs or, for all $p \in \N$,   in the class of graphs of
treedepth at most~$p$. In accordance with the  notation in
Lemma~\ref{lemma:ProtInducedSubgraph}, the class $\cal F$ corresponds to the
class where $\Pi$ has FII. The choice of our parameter now ensures that our
kernelization replaces protrusions of treedepth at most a previously fixed
constant~$d$: choosing $\cal G$ to be the graphs of treepdepth at most $d$,
all protrusions (actually the graphs induced by them)  are members of $\cal F
\cap \cal G$. As $\cal G$ is well-quasi ordered under the label-preserving induced subgraph
relation~\cite[Chapter~6, Lemma~6.13]{NOdM12},  we choose $\preceq$ to be
$\subseteq_{\text{ind}}$. 

Now consider a restriction of the graph problem~$\Pi$ to a class $\cal K$ that 
is closed under taking induced subgraphs. In this paper, the class $\cal K$ 
is a hereditary graph class of bounded expansion or 
a hereditary and nowhere dense class.  
This ensures that
$\emptyset \neq \cal K \cap \cal G \subseteq \cal F \cap \cal G$. 
Given an instance $(G, \probpar)$ of~$\Pi$ with $G \in \cal K$, 
one can replace a protrusion of $G$ by a representative (of constant size) 
that is an \emph{induced subgraph} of that protrusion, ensuring that
this replacement creates a graph that still resides in $\cal K$. To summarize,
Lemma~\ref{lemma:ProtInducedSubgraph} guarantees that the protrusion 
replacement rule (described next) preserves the graph class $\cal K$ and
the parameter.

As preparation for the kernelization theorems of the next section,
let $\PPP$ denote the set of all graph problems that have FII
on general graphs or, for each~$p \in \N$, in the class of graphs of treedepth at most~$p$. 
Our reduction rule is formalized as follows.
\begin{redrule}[Protrusion replacement]\label{rrule:Protrusion}
    Let $t,d \in \N$ and let \mbox{$\Pi \in \PPP$}. 
    Let $\Rep$ be a class of boundaried graphs of treedepth at most
    $d$ containing representatives of the equivalence classes of
    $\equiv_{\Pi,i}$ restricted to the graphs of treedepth at most~$d$, for $i=1,\dots,t$.
    Let $(G,\probpar)$ be an instance of $\Pi$ and assume that 
    $W \subseteq V(G)$ is a $t$-protrusion of treedepth at most~$d$ and
    boundary size $i=|\partial_G(W)|\leq t$.
    Let $\widetilde{R} \in \Rep$ be a $\equiv_{\Pi, i}$-representative of~$\widetilde{G}[W]$. 
    The protrusion replacement rule is the following:
$$
    \mbox{Reduce $(G,\probpar)$ to~ $(G',\probpar') := \big((G \ominus G[W]) 
    \oplus \widetilde{R},\> \probpar + \Delta_{\Pi,i}(\widetilde{G}[W],\widetilde{R})\big)$.}
$$
\end{redrule}

\noindent
We let $\cal F$ denote the class on which the problem has FII
and by $\cal G$ the class of graphs of treedepth at most~$d$.
The existence of a suitable {\em finite set of representatives} $\cal R(t,d)$
for Rule~\ref{rrule:Protrusion} is guaranteed by Lemma~\ref{lemma:ProtInducedSubgraph}:
we let $\Rep$ denote the finite set
$\bigcup_{i=1}^t \cal R(i, \cal F, \cal G, \subseteq_{\text{ind}})$ from
Lemma~\ref{lemma:ProtInducedSubgraph},
and $\prot$ denote the size of the largest member of $\Rep$.
The safety of the protrusion replacement
follows from the definition of FII. 

In what follows, when applying the protrusion replacement by
Rule~\ref{rrule:Protrusion}, we will always assume that for each $t,d \in \N$, 
we are given the finite set $\Rep$ of representatives.
Note that previous work on meta-kernels implicitly made this 
assumption~\cite{BFLPST09,FLST10,FLMPS11,FLMS12}.

\section{Linear Kernels on Graphs of Bounded Expansion}\label{sec:Kernels}

In this section we show that graph problems that have finite integer index on
general graphs or in the class of graphs with bounded treedepth  admit linear
kernels on hereditary graph classes with bounded expansion, when parameterized
by the size of a modulator to constant treedepth. On nowhere dense classes,
we obtain an almost-linear kernel. Our main theorem is the
following.
\begin{theorem}\label{theorem:KernelBoundedExp}
        Let $\cal K$ be a graph class of bounded expansion and let~$d \in \N$ be a constant.
        Let~$\Pi \in \PPP$. Then there is an algorithm that takes as input 
	$(G,\probpar) \in \cal K \times \N$ 
        and, in time $\bigO(|G| + \log \xi)$, outputs $(G',\probpar')$ such that
        \begin{enumerate}
                \item $(G,\probpar) \in \Pi$ if and only if $(G',\probpar') \in \Pi$;
                \item $G'$ is an induced subgraph of $G$; and
                \item $|G'| = \bigO(|S|)$, where~$S$ is an optimal treedepth-$d$
                        modulator of the graph~$G$.
        \end{enumerate}
\end{theorem}

\noindent
In the rest of this section we fix a problem $\Pi \in \frak{P}$ and let $\cal K$ be a 
hereditary graph class whose expansion is bounded by a function $f \colon \N
\rightarrow \R$.

We proceed as follows. Because an optimal treedepth-$d$ modulator cannot be assumed as 
part of the input, we obtain an approximate modulator $S
\subseteq V(G)$ to partition~$V(G)$ into sets $\YYYY$ such
that $S \subseteq Y_0$ and $|Y_0| = \bigO(|S|)$ and for $1 \leq i \leq l$,
$Y_i$ induces a collection of connected components of $G - Y_0$ that have exactly the
same \emph{small} neighborhood in~$Y_0$. We then use bounded expansion to 
show that $\ell = \bigO(|S|)$ and use protrusion reduction 
to replace each~$G[Y_i]$, $1 \leq i \leq l$, by an \emph{induced subgraph} of $G[Y_i]$ of 
constant size. Every time the protrusion replacement rule is applied, $\probpar$ is
modified. This results in an equivalent instance $(G',\probpar')$ such that
$G' \subseteq G$ and $|G'| = \bigO(|S|)$, as claimed in Theorem~\ref{theorem:KernelBoundedExp}.

\begin{lemma}\label{lemma:TDModApprox}
  Fix~$d \in \N$. Given a graph~$G$, one can in $O(|G|^2)$ time compute
  a subset $S \subseteq V(G)$ such that $\td(G-S) \leq d$ and $|S|$ is
  at most $2^d$ times the size of an optimal treedepth-$d$ modulator
  of~$G$. For graphs of bounded expansion, the set~$S$ can be computed in
  linear time. For nowhere dense classes it can be computed in time~$O(|G|^{1+\varepsilon})$
  for every fixed~$\varepsilon > 0$.
\end{lemma}
\begin{proof}
	We use the fact that any DFS-tree of a graph of treedepth~$d$ has
	height at most~$2^d-1$. We compute a DFS-tree 
	of the graph~$G$ and if it has height more than $2^d-1$, then $\td(G)>d$.
        So, we take some path $P$ from the root of the tree of length $2^d-1$ and add all the
	$2^d$ vertices of $P$ into a set $S_0\subseteq S$; delete $V(P)$ from the graph and repeat.
	(Clearly, at least one of the vertices of $P$ must be in any modulator.) At the end 
	of this procedure, the DFS-tree of the remaining graph $G-S_0$ has height at most $2^d - 1$. 
	This gives us a tree (path) decomposition of the graph of width at most~$2^d - 2$. 
	Now use standard tools (e.g., Courcelle's theorem \cite{Cou90}) to obtain an optimal treedepth-$d$ modulator
	$S_1$ in $G-S_0$, and set $S=S_0\cup S_1$. Since the 
	treewidth of $G-S_0$ is a constant, the latter algorithm runs in 
	time linear in the size of the graph. The overall size of the modulator
	is at most $2^d$ times the optimal solution.  

  For a graph $G$ from a class of bounded expansion, we modify the way $S_0$
  is computed above (the resulting set will not be larger than the one
  computed above, and often much smaller). By~\cite{NOdM06}, graph classes of
  bounded expansion admit low treedepth coloring: Given any integer~$p$, there
  exists an integer~$n_p$ such that any graph of the class can be properly
  vertex colored using~$n_p$ colors such that for any set of $1 \leq i \leq p$
  colors, the graph induced by the vertices that receive these $i$ colors has
  treedepth at most~$i$. Such a coloring is called a $p$-treedepth coloring
  and can be computed in linear time \cite{NOdM06}. 
        Here we choose $p = 2^d$ and obtain such a coloring for~$G$ using~$n_p$ colors. 
	Let~$G_1, \ldots, G_r$ denote the subgraphs induced by at most~$2^d$
	of these color classes where~$r < 2^{n_p}=\bigO(1)$.  
	Note that $\sum_j |G_j| = \bigO(|G|)$, since every vertex of~$G$ appears in at most a constant number of subgraphs. 
	Any path in~$G$ of length $2^d-1$ must be in some subgraph~$G_j$, for $1 \leq j \leq
	r$, and we hit all such paths with a set $S_0$ obtained in the
	following iterated procedure.

	Start with $S_0=\emptyset$. For each $j=1,2,\dots,r$,
	we simply construct a treedepth decomposition of $G_j-S_0$, e.g., by the depth-first search.
	Using standard dynamic programming we find an optimum hitting set
	for the set of all paths of length~$2^d-1$ in $G_j-S_0$ and add its vertices into $S_0$ (and delete them from the graph). 
	Again, some hitting set for these paths must be in any modulator.
	The time taken to do this for each subgraph~$G_j-S_0$ is $\bigO(|G_j|)$. 
	The total time taken is therefore $\sum_j |G_j| = \bigO(|G|)$.

  The approach for nowhere dense classes is nearly the same: by~\cite{NOdM08,NOdM12},
  for a nowhere dense class~$\cal G$ and $\varepsilon' > 0$, $p \in \N$ there
  exists a threshold~$N_{\varepsilon',p}$ such that for all~$G \in \cal G$ with
  $|G| \geq N_{\varepsilon',p}$ it holds that~$G$ has a $p$-treedepth coloring
  with at most~$|G|^{\varepsilon'}$ colors. Therefore, for every $\varepsilon
  > 0$, the above algorithm runs in time~$O(|G|^{1+\varepsilon})$ by
  choosing~$\varepsilon' = \varepsilon/p$ and~$p = 2^d$; now the
  subgraphs~$G_1,\ldots, G_r$ induced by at most~$2^d$ colors have again
  treedepth at most~$2^d$ while~$r \leq (|G|^{\varepsilon'})^p =
  |G|^{\varepsilon}$. The running time to construct~$S_0$ is, accordingly,
  $\sum_j |G_j| = \bigO( |G|^{1+\varepsilon})$ and this also bounds the total
  running time.
\end{proof}

\noindent
We will make heavy use of the following lemma to prove the kernel size.

\begin{lemma}\label{lemma:BipartiteNab1}
	Let $G=(X,Y,E)$ be a bipartite graph, and $p\geq\grad_1(G)$. Then there are at most  
	\begin{enumerate}
		\item $2p \cdot |X|$ vertices in~$Y$ with degree greater than $2p$;
                \item $(4^{p} + 2p) \cdot |X|$ subsets~$X' \subseteq X$ such that 
			$X' = N(u)$ for some $u \in Y$.
	\end{enumerate}
\end{lemma}
\begin{proof}
	We construct a sequence of graphs $G_0,G_1,\dots,G_\ell$ such that 
	$G_i \in G \nab 1$ for all $0 \leq i \leq \ell$ as follows. 
	Set $G_0 = G$, and for $0 \leq i \leq \ell-1$ construct $G_{i+1}$ 
	from $G_i$ by choosing a vertex $v \in V(G_i) \setminus X$ such 
	that $N(v) \subseteq X$ contains two non-adjacent vertices $u,w$ in $G_i$; 
	and contract this edge to the vertex~$u$ to obtain $G_{i+1}$. Recall that 
	contracting $uv$ to $u$ is equivalent to deleting vertex $v$ and adding 
	edges between each vertex in $N(v) \setminus u$ and $u$. It is clear 
	from the construction that for $0 \leq i \leq \ell$, 
	$X \subseteq V(G_i) \subseteq X \cup Y$.

	This process clearly terminates, as $G_{i+1}$ has at least one more 
	edge between vertices of $X$ than $G_i$. Note that $G_i \in G \nab 1$ 
	for $0 \leq i \leq \ell$, as the edges $e_1,\dots,e_{i-1}$ that were 
	contracted to vertices in $X$ in order to construct $G_i$ had one 
	endpoint each in $X$ and $Y$, the endpoint in $Y$ being deleted after 
	each contraction. Thus, $e_1, \dots, e_{i-1}$ induce a set of stars
	in $V(G) = V(G_0)$, and $G_i$ is obtained from $G$ by contracting these stars. 
	We therefore conclude that $G_i$ is a depth-one shallow minor of $G$. 
	In particular, this implies $G_\ell[X]$ is $2p$-degenerate and 
	has at most $2p|X|$ edges.
	Further, note that for each $0 \leq i \leq \ell$, $Y \cap V(G_i)$ is, 
	by construction, still an independent set in~$G_i$.
	
	Let us now prove the first claim. To this end, assume that there is
	a vertex $v \in Y\cap V(G_\ell)$ such that $degree(v) > 2p$. 
	We claim that $G_\ell[N(v)]$ (where  $N(v) \subseteq X$) is a clique. 
	If not, we could choose a pair of non-adjacent vertices in 
	$G_\ell[N(v)]$ and construct a ($\ell+1$)-th graph
	for the sequence which would contradict the fact that $G_{\ell}$ 
	is the last graph of the sequence. 
	However, a clique of size $|\{v\}\cup N(v)| > 2p+1$ 
	is not $2p$-degenerate.
	Hence we conclude that no vertex of $Y \cap V(G_{\ell})$ has 
	degree larger than $2p$ in $G_\ell$ 
	(and in $G$). Therefore the vertices of $Y$ of degree greater 
	than $2p$ in the graph~$G$, if there were any, 
	must have been deleted during the edge contractions that resulted 
	in the graph $G_{\ell}$. As every contraction added at least 
	one edge between vertices in $X$ and since $G_\ell[X]$ contains at most 
	$2p|X|$ edges, the first claim follows.

	For the second claim, consider the set $Y' = Y \cap V(G_\ell)$. 
	The neighborhood of every vertex $v \in Y'$ induces a clique in 
	$G_\ell[X]$. From the $2p$-degeneracy of $G_\ell[X]$ and
	Proposition~\ref{prop:DegenerateNumCliques}, it follows
	that $G_\ell[X]$ has at most 
	$2^{2p}|G_\ell[X]|=4^{p} \cdot |X|$ cliques. Thus the 
	number of subsets of $X$ that are neighborhoods of vertices in $Y$ in $G$ is
	at most $(4^p+2p) \cdot |X|$, where we accounted 
	for vertices of $Y$ lost via contractions by the bound on the number 
	of edges in $G_\ell[X]$.
\end{proof}

\noindent
The following two corollaries to Lemma~\ref{lemma:BipartiteNab1} show how it
can be applied in our situation.

\begin{corollary}\label{cor:LargeDegree}
	Let $\cal K$ be a graph class whose expansion is bounded by a 
	function $f \colon \N \rightarrow \R$. Suppose that for 
	$G \in \cal K$ and $S \subseteq V(G)$, $C_1, \ldots, C_s$ are 
	disjoint connected subgraphs of $G-S$ satisfying the following 
	two conditions
    \begin{enumerate}\item
	for $1 \leq i \leq s$, $\diam(G[V(C_i)]) \leq \delta$ and
	\item $|N_S(C_i)| > 2 \cdot f(\delta+1)$.
    \end{enumerate}
	Then $s \leq 2 \cdot f(\delta+1) \cdot |S|$.
\end{corollary} 
\begin{proof}
	We construct an auxilliary  bipartite graph $\bar G$ with partite 
	sets $S$ and $Y = \{C_1,\ldots,C_s\}$. There is an edge between $C_i$ 
	and $x \in S$ iff $x \in N_S(C_i)$. Note that $\bar G$ is a 
	depth-$\delta$ shallow minor of $G$ with branch sets $C_i, 1 \leq i \leq s$.
	In relation to Lemma~\ref{lemma:BipartiteNab1} we would like to show
	that, for any $F\in\bar G\nab1$, it is $F\in G\nab(\delta+1)$
	(while $\nab$ is not additive in general).
	This follows since a branch set of $F$ in $G$ is induced by a vertex of $S$
	plus a subcollection of attached sets $C_i$, $1 \leq i \leq s$, or by one set
	$C_i$ and a subset of attached vertices from~$S$.
	In both the cases the radius is at most
	$1+\max_{i}\diam(C_i)\leq\delta+1$.
	
	Consequently,
	$\grad_1(\bar G)\leq\grad_{\delta+1}(G)\leq f(\delta+1)$
	and, by Lemma~\ref{lemma:BipartiteNab1} for the choice $p=f(\delta+1)$,
	$$s\leq 2p|S|=2f(\delta+1)\cdot|S| .\vspace*{-3ex}$$
\end{proof}

\begin{corollary}\label{cor:NumClusters}
	Let $\cal K$ be a graph class whose expansion is bounded by a 
	function $f \colon \N \rightarrow \R$. Suppose that for 
	$G \in \cal K$ and $S \subseteq V(G)$, $\mathcal C_1, \ldots, \mathcal C_t$ 
	are sets of connected components of $G-S$ such that for all pairs
	$C,C' \in \bigcup_i \mathcal{C}_i$ it holds that $C, C' \in \mathcal C_j$ 
	for some $j$ if and only if $N_S(C) = N_S(C')$. Let $\delta > 0$ 
	be a bound on the diameter of the components, \ie 
	for all $C \in \bigcup_i \mathcal{C}_i$, $\diam(G[V(C)]) \leq \delta$.
	Then there can be only at most 
	$t\leq(4^{f(\delta+1)} + 2f(\delta+1)) \cdot |S|$ such sets~$\mathcal{C}_i$.
\end{corollary}
\begin{proof}
	As in the proof of Corollary~\ref{cor:LargeDegree}, we construct 
	a bipartite graph $\bar G$ with partite sets $S$ and 
	$Y = \{C_1, \ldots, C_r\}$, and argue about $\grad_1(\bar
	 G)\leq\grad_{\delta+1}(G)\leq f(\delta+1)$.
 	By Lemma~\ref{lemma:BipartiteNab1}, for $p=f(\delta+1)$,
	$$t\leq (4^{p}+2p)|S|=
	 (4^{f(\delta+1)}+2f(\delta +1)) \cdot |S| .\vspace*{-3ex}$$
\end{proof}

\noindent
In the first phase, our kernelization algorithm partitions an input graph
according to a low-treedepth modulator (as found in
Lemma~\ref{lemma:TDModApprox}).

\begin{lemma}\label{lemma:Decomposition}
	Let~$\mathcal{K}$ be a graph class with expansion bounded by~$f$,
	$G \in \mathcal{K}$ and $S \subseteq V(G)$ be a treedepth-$d$ modulator	($d$ a constant).
	There is an algorithm that runs in time $\bigO(|G|)$ 
	and partitions $V(G)$ into sets $\YYYY$ such that the following hold:
	\begin{enumerate}
		\item $S \subseteq Y_0$ and $|Y_0| = \bigO(|S|)$; 
		\item for $1 \leq i \leq \ell$, $Y_i$ induces a set of connected 
			components of $G-Y_0$ that have 
		the same neighborhood in $Y_0$ of size at most $2^{d+1} + 2\cdot f(2^d)$;
		\item $\ell\leq \big(4^{f(2^d)}+2f(2^d)\big) \cdot |S| = \bigO(|S|)$.
	\end{enumerate}
\end{lemma}
\begin{algorithm}[t]
        \small
        \caption{{\label{fig:MarkingAlg}\sc Bag marking algorithm}}
        \KwIn{A graph $G$, a subset $S \subseteq V(G)$ such that
        $\td (G-S) \leqslant d$, and an integer $t > 0$.}
        \BlankLine

        Set $\mc{M} \leftarrow \emptyset$ as the set of marked bags\;

        \For{each connected component~$C$ of $G-S$ such that $N_S(C) \geq t$}
        {Choose an arbitrary vertex $v \in V(C)$ as a root and construct a
        DFS-tree starting at~$v$\;

        Use the DFS-tree to obtain a path-decomposition $\mc{P}_C = (P_C, \mc{B}_C)$
        of width at most $2^{d} - 2$ in which the bags are ordered from left to right\;
        }

        \BlankLine

        Repeat the following loop for the path-decomposition $\mc{P}_C$ of every $C$\;
        \While{$\mc{P}_C$ contains an unprocessed bag}{

        \BlankLine

        Let $B$ be the leftmost unprocessed bag of $\mc{P}_C$\;

        Let $G_B$ denote the subgraph of $G$ induced by the vertices in the bag~$B$
        and in all bags to the left of it in $\mc{P}_C$.

        \BlankLine

        \textbf{[Large-subgraph marking step]}\\
        \If{$G_B$ contains a connected component $C_B$ such that $|N_{S}(C_B)|\geqslant t$}{
        $\mc{M} \leftarrow\mc{M} \cup \{B\}$ and remove the vertices of $B$ from every bag
        of $\mc{P}_C$\;
        }
        \BlankLine Bag $B$ is now processed\; }
        \BlankLine
        \Return{$Y_0 = S \cup V(\mc{M})$}\;
		\caption{{\sc Bag marking algorithm}} \label{fig:MarkingAlg}
\end{algorithm}

\begin{proof}
	We first construct a DFS-forest~$\cal{D}$ of $G-S$. Assume that there are $q$ 
	trees $T_1, \ldots, T_q$ in this forest rooted at $r_1, \ldots, r_q$, respectively. 
	Since $\td(G-S) \leq d$, the height of every tree in~$\cal{D}$ is at most~$2^{d}-1$. 
	Next we construct for each $T_i$, $1 \leq i \leq q$, a path decomposition of the 
	subgraph of $G[V(T_i)]$ as follows. Suppose that $T_i$ has leaves $l_1, \ldots, l_s$ 
	ordered according to their DFS-number. For $1 \leq j \leq s$, create a bag~$B_j$ 
	containing the vertices on the unique path from $l_j$ to $r_i$ and string 
	these bags together in the order $B_1, \ldots, B_s$. Clearly, this is a path 
	decomposition~$\mathcal{P}_i$ of $G[V(T_i)]$ with width at most $2^d-2$. 
	Note that the root~$r_i$ is in every bag of~$\cal P_i$. These first steps
	are depicted in the first loop of Algorithm~\ref{fig:MarkingAlg} and
	clearly run in linear time.

	We now use a marking algorithm similar to the one in~\cite{KLPRRSS12} to 
	mark $\bigO(|S|)$ bags in the path decompositions $\mathcal{P}_1, \ldots, \mathcal{P}_q$ 
	with the property that each marked bag can be uniquely identified with a 
	connected subgraph of $G-S$ that has a large neighborhood in the modulator~$S$. 

	This algorithm is described in Figure~\ref{fig:MarkingAlg} in which we 
	set $t$, the size of a {\em large neighborhood} in $S$, to be $t:= 2\cdot f(2^d)+ 1$. 
	Note that there is a one-to-one correspondence between marked bags $\mathcal{M}$ 
	and connected subgraphs with a neighborhood of size at least~$t$ in $S$.  
	Moreover each connected subgraph has treedepth at most~$d$ and hence diameter 
	at most~$2^d-1$. By Corollary~\ref{cor:LargeDegree}, the number of connected subgraphs
	of large neighborhood and hence the number of marked bags 
	is at most $2 \cdot f(2^d-1+1) \cdot |S| = \bigO(|S|)$. 
	We set $Y_0 := V(\mathcal{M}) \cup S$. As the marking stage of the algorithm
	runs through every path-decomposition of the components of $G-S$ exactly once, 
	this phase takes only linear time.

	Now observe that each connected component in $G - Y_0$ has less than 
	$t = 2 \cdot f(2^d) + 1$ neighbors in~$S$: for every connected 
	subgraph~$C$ with at least~$t$ neighbors in~$S$, there exists a marked 
	bag~$B$. Importantly, the bag~$B$ was the \emph{first} bag that was marked before 
    the number of neighbors in~$S$ of \emph{any} connected subgraph reached the threshold~$t$. 
    Hence each connected component  of $G[V(C) \setminus B]$  has degree less than $t$ in~$S$. 
    Since every component is connected to at most two marked bags (in $Y_0$) and 
	since each bag is of size at most~$2^d-1$, the size of the neighborhood of 
	every component of $G-Y_0$ in~$Y_0$ is at most $2(2^d-1)+t \leq 2^{d+1} + 2 \cdot f(2^d)$. 
	
	To complete the proof, we simply cluster the connected components of $G-Y_0$ 
	according to their neighborhoods in $Y_0$ to obtain the sets $Y_1, \ldots,Y_\ell$. 
	Since each connected component of $G-S$ is of diameter $\delta\leq 2^d-1$,
	by Corollary~\ref{cor:NumClusters}, the number~$\ell$ of clusters is at most 
	$\big(4^{f(2^d)}+2f(2^d)\big) \cdot |S| = \bigO(|S|)$, as claimed.

	To accomplish this feat in linear time, we assume an arbitrary order
	of the vertices in $Y_0$ (say, the order in which they appear in the encoding of
	the graph). A simple bipartite auxillary graph with partitions $Y_0$ and $V(G)\setminus Y_0$
	with edge set $(Y_0 \times V(G)\setminus Y_0) \cap E(G)$ can be used to find the
	neighbors of a vertex $v \not \in Y_0$ inside $Y_0$ in constant time as the number of 
	neighbors of such a vertex is at most $2^{d+1}+2 \cdot f(2^d)=\bigO(1)$. Thus, computing the
	neighbors in $Y_0$ of every component of $G - Y_0$ takes only linear time. 
	If we store the neighborhoods of these components as lists sorted according to the
	ordering of $Y_0$ inside an array of length $O(|S|)$, we can sort this array in linear time
	using bucket sort: each entry of the list has encoding length at most $\log |V(G)|$, therefore 
	they can be compared in constant time under the usual RAM-model. The clusters can then
	be simply read from the array. Thus the clustering of the components of $G-Y_0$ and therefore 
	the whole decomposition is linear-time computable.
\end{proof}

\noindent
To prove a linear kernel, all that is left to show is that each cluster $Y_i$, $1 \leq i \leq \ell$, can be reduced to constant 
size. Note that each cluster is separated from the rest of the graph via a small set of vertices in $Y_0$ and that each 
component of $G-Y_0$ has constant treedepth even when its bounday is included. 
These facts enable us to use the protrusion reduction rule.

In the proof of the following lemma it will be convenient to use the
following normal form of tree decompositions: A triple $(T, \{ W_x \mid x
\in V(T) \}, r)$ is a \emph{nice tree decomposition} of a graph $G$ if $(T,
\{ W_x \mid x \in V(T) \})$ is a tree decomposition of~$G$, the tree $T$ is
rooted at node $r \in V(T)$, and each node of $T$ is of one of the following
four types:
\begin{enumerate}\parskip0pt
	\item a \emph{leaf node}: a node having no children and containing
	  exactly one vertex in its bag;
	\item a \emph{join node}: a node $x$ having exactly two children
	  $y_1,y_2$, and $W_x=W_{y_1}=W_{y_2}$;
	\item an \emph{introduce node}: a node $x$ having exactly one child
	  $y$, and $W_x=W_y\cup \{v\}$ for a vertex $v$ of $G$ with $v \not \in W_y$;
	\item a \emph{forget node}: a node $x$ having exactly one child $y$,
	  and $W_x=W_y\setminus \{v\}$ for a vertex $v$ of $G$ with $v \in W_y$.
\end{enumerate}
Given a tree decomposition of a graph $G$ of width $w$, one can
effectively obtain in time $\bigO(|V(G)|)$ a nice tree decomposition of $G$ with
$\bigO(|V(G)|)$ nodes and of width at most~$w$~\cite{BK91}.

For the next statement, recall our fixed problem $\Pi \in \frak{P}$, the (implicitly given) finite set
$\Rep$ of representatives of the equivalence classes of the relations $\equiv_{\Pi,i}$,
$i=1,\dots,t$, restricted graphs of treedepth $\leq d$, 
and $\prot$ the size of the largest member(s) of $\Rep$.

\begin{lemma}\label{lemma:ProtrusionReplacement}
	For fixed~$d,h \in \N$ (constants) and $\cal K$ a hereditary graph class,
	let~$(G,\probpar)$ be an instance of~$\Pi$ 
        with $G \in\cal K$ and let~$S \subseteq V(G)$ be a treedepth-$d$ modulator of~$G$. 
	Let \YYYY be a vertex partition of $V(G)$ such that
    \begin{itemize}\parskip 1pt
	\item $S \subseteq Y_0$,
	\item $N(Y_i)\subseteq Y_0$ and $|N_{Y_0}(Y_i)| \leq h$
	 for $1 \leq i \leq \ell$, and
	\item $N_{Y_0}(Y_i)\not=N_{Y_0}(Y_j)$ for $i\not=j$.
    \end{itemize}
	Then one can in $\bigO(|G|+\log \probpar)$ time obtain an instance
	$(G',\probpar')$ and a vertex partition \primedYYYY of~$V(G')$ such that
	\begin{enumerate}\parskip0pt
		\item $(G,\probpar) \in \Pi$ if and only if $(G',\probpar') \in \Pi$;
		\item $G' \in\cal K$ is an induced subgraph of $G$ with $Y_0' = Y_0$;
		\item for $1 \leq i \leq \ell$ it is $|N_{Y_{0}'}(Y_i')| \leq h$, 
			$\td(G[Y_i']) \leq d$, and\\
			$|Y_i'| \leq \rho(d+h,d+h) = \bigO(1)$.
	\end{enumerate}	
\end{lemma}
\begin{proof}
	Since $S \subseteq Y_0$ is a treedepth-$d$ modulator, for all $1 \leq i \leq \ell$, 
	we have $\td(G[Y_i]) \leq d$ and hence $\tw(G[Y_i]) \leq d-1$. 
	Moreover treedepth at most~$d$ implies diameter at most~$2^d-1$ for each component.  
	Since $Y_0' = Y_0$, let $N(X)$ stand for $N_{Y_0}(X)=N_{Y_0'}(X)$.
	For each index~$1 \leq i \leq \ell$, our algorithm constructs a tree-decomposition 
	of $G[Y_i \cup N(Y_i)]$ of width~$d-1+h$ that satisfies certain properties 
        that we mention below.

	The algorithm then uses the tree-decomposition to 
	replace~$Y_i$ in a systematic manner using the protrusion replacement
	Rule~\ref{rrule:Protrusion}. 
	The special properties of the tree-decomposition enable the algorithm 
	to perform this replacement in $\bigO(|Y_i \cup N(Y_i)|)$ time. Total 
	time used to replace all sets~$Y_i$ is $\sum_{i=1}^{\ell}|Y_i \cup N(Y_i)|$.
	Since, by Corollary~\ref{cor:NumClusters} (with $Y_0$ in the place of $S$), 
	$\sum_{i=1}^{\ell}|N(Y_i)| = \bigO(\ell) = \bigO(|Y_0|)$, 
	the running time is indeed $\bigO(|G|)$. It therefore suffices to
	specify	the properties of our tree-decompositions and describe how each~$Y_i$ 
	is replaced with~$Y_i'$.

\smallskip
	The desired tree-decomposition $\mathcal{T}_i = \big(T_i, \{W_x \mid x \in V(T_i)\}\big)$ 
	of width at most~$d+h-1$ for $G[Y_i \cup N(Y_i)]$ satisfies the following conditions: 
	\begin{enumerate}\parskip0pt
		\item there is a node~$r \in V(T_i)$ such that $W_r = N(Y_i)$;
		\item the tree-decomposition is nice and the leaf bags contain one vertex. 
	\end{enumerate}  
	The first condition can be achieved by simply modifying the graph $G_i$ so that 
	$N(Y_i)$ induces a clique, and then introducing an extra node $r$
	if no such node exists. The decomposition~$\mathcal{T}_i$ is rooted at the node~$r$.
	For~$x \in V(T_i)$, we let~$\widetilde{G}_x$ denote the $t$-boundaried 
	graph induced by the vertices in the bags of the subtree of~$T_i$
	rooted at~$x$. That is, formally,
	\[
		{G}_x := G \left [ 
			W_x\cup \bigcup\nolimits_{y\mbox{ \scriptsize 
				descendant of }x} W_y \right ]
		\mbox{ and~ }
		\widetilde{G}_x := \left( G_x, W_x \right),
	\]    
	where the boundary $\bound(G_x) = W_x$ is of size $t\leq d+h$. 
	Then ${G}_r = G[Y_i \cup N(Y_i)]$. 
	Note that the treedepth of $G_x$ is at most~$d+|W_x\cap S|\leq d+h$.

	Recall that $\Pi$ has \FII either on general graphs or on bounded treedepth graphs. 
        Using Lemma~\ref{lemma:ProtInducedSubgraph}, for each $x \in V(T_i)$, 
	there exists a representative $\Lambda(x)\in \mathcal{R}(d+h,d+h)$
	of $\widetilde{G}_x$ which is an induced subgraph of $G_x$
	and $\bound(\Lambda(x)) = \bound(G_x)$.
	Replacing $\widetilde{G}_x$ by $\Lambda(x)$ hence does not increase the treedepth. 
	Furthermore, $|\Lambda(x)|\leq M:=\rho(d+h,d+h)$ which is a constant.
	Let $\mu(x) = \Delta_{t}(\widetilde{G}_x,\Lambda(x))$. 

	Our task is to find~$\Lambda(r)$ and $\mu(r)$ which we will calculate 
	in a bottom-up manner along $T_i$ in $\bigO(|Y_i|)$ 
	time as follows. If $y \in V(T_i)$ is a leaf node then these values
	can be computed in constant time. Let $x \in V(T_i)$ be a node with 
	exactly one child $y$ whose $\Lambda$ and $\mu$ values are known. 
	Consider the $t$-boundaried graph $\widetilde{G}'_x$ where $t\leq d+h$ and
	$${G'_x} := (G_x \ominus_{W_y} G_y) \oplus_{W_y} \Lambda(y)
	\mbox{ with } \bound(G'_x) = W_x .$$
	We claim that 
	$\widetilde{G}'_x \equiv_{t} \widetilde{G}_x$. To prove this, we need to 
	demonstrate that for all $t$-boundaried graphs $\widetilde{G}$ and all $\probpar \in \N$,
	\[(\widetilde{G}'_x \oplus_{W_x} \widetilde{G}, \probpar) \in \Pi \ \text{if and only if} \ 
		(\widetilde{G}_x \oplus_{W_x} \widetilde{G}, \probpar - \mu') \in \Pi,
	\]
	where~$\mu' = \Delta_{t}(\widetilde{G}_x,\widetilde{G}'_x)$ is to be specified. Now
	\begin{align*}
		(\widetilde{G}'_x \oplus_{W_x} \widetilde{G}, \probpar) \in \Pi 
			\ & \text{iff} \ ((\widetilde{G}_x \ominus_{W_y} G_y) \oplus_{W_y} \Lambda(y)) \oplus_{W_x}\widetilde{G}, \probpar) \in \Pi \\
			\ & \text{iff} \ ((\widetilde{G}_x \oplus_{W_x}\widetilde{G}) \ominus_{W_y} G_y) \oplus_{W_y} \Lambda(y), \probpar) \in \Pi \\
			\ & \text{iff} \ ((\widetilde{G}_x \oplus_{W_x}\widetilde{G}) \ominus_{W_y} G_y) \oplus_{W_y} \widetilde{G}_y, \probpar - \mu(y)) \in \Pi, 
	\end{align*}
	where the last step follows because of $\Lambda(y) \equiv_{t} \widetilde{G}_y$. 
	Since 
	$$(\widetilde{G}_x \oplus_{W_x} \widetilde{G}) \ominus_{W_y} G_y) \oplus_{W_y} \widetilde{G}_y
	 \>=\> \widetilde{G}_x \oplus_{W_x} \widetilde{G} ,$$
	this proves our claim.
	In fact, $\mu'=\mu(y)$.

	Observe that $G_x'$ is of \emph{constant} size, bounded from 
	above by $M+|W_x|\leq M+d+h=\bigO(1)$. Since $\Lambda(y)$ is an induced subgraph of $G_y$,
	it follows that $G_x'$ is an induced subgraph of $G_x$ and therefore has treedepth 
	at most~$d+h$.
	Then we can find in constant time the associated representative
	$\widetilde{R}\in\mathcal{R}(d+h,d+h)$ of $\widetilde{G}'_x$.
	We set $\Lambda(x) := \widetilde{R}$ and $\mu(x) := \mu' +
	 \Delta_{t}(\widetilde{G}'_x,\widetilde{R})$. 
	Note that the total time spent at node~$x$ to generate these values is a constant.  	

	Lastly, consider the case when $x \in V(T_i)$ has exactly two 
	children~$y_1$ and~$y_2$ whose $\Lambda$ and $\mu$ values are known. 
	Since our tree-decomposition is nice, we have $W_{y_1} = W_x = W_{y_2}$ and 
	therefore $\bound(G_{y_1}) = \bound(G_{y_2}) = W_x$. Take the $t$-bound\-aried graph 
	$\widetilde{G}_x''$ where $t\leq d+h$ and
	$${G_x''} := \Lambda(y_1) \oplus_{W_x} \Lambda(y_2) \mbox{ with }
		\bound(G_x'') = W_x .$$
	Similarly as in the previous case, one can show that for all graphs 
	$\widetilde{G}$ and all $\probpar \in \N$,
	\[(\widetilde{G}_x'' \oplus_{W_x} \widetilde{G}, \probpar) \in \Pi \ \text{if and only if} \ 
		(\widetilde{G}_x \oplus_{W_x} \widetilde{G}, \probpar - \mu'') \in \Pi, 
	\]
	where  $\mu'' = \mu(y_1)+\mu(y_2)$.
	The graph ${G_x''}$ has size at most $2M$ which is a constant. 
	One can therefore, again in constant time, calculate a 
	representative $\widetilde{R} \in \mathcal{R}(d+h,d+h)$ of $\widetilde{G}_x''$. 
	We set $\Lambda(x) := \widetilde{R}$ and $\mu(x) := \mu''+
	 \Delta_{d+h}(\widetilde{G}_x'',\widetilde{R})$.

\smallskip
	To summarize, our proof shows that one can, independently for each
	$i\in\{1,\dots,\ell\}$,	in time $\bigO(|T_i|)=\bigO(|Y_i|)$ obtain 
	$\Lambda(r)$ and $\mu(r)$ (where $r$ is the root of the
	tree-decomposition $\mathcal{T}_i$ for $G[Y_i \cup N(Y_i)]$) 
	with the following properties:
	for all graphs $\widetilde{G}$ and all $\probpar \in \N$,
	$$(\widetilde{G}_r \oplus \widetilde{G}, \probpar) \in \Pi
	\mbox{ if and only if }
	(\Lambda(r) \oplus \widetilde{G}, \probpar+\mu(r)) \in\Pi .$$
	Let $\mu_i:=\mu(r)$ and
	$Y'_i:=V(\Lambda(r))\setminus Y_0$ be the chosen replacement of the cluster $Y_i$.
	Then $G[Y'_i]$ is an induced subgraph of $G[Y_i]$ of constant size, and the
	neighborhood of $Y_i'$ inside $Y_0$ is untouched. 
	It immediately follows that $\td(G[Y'_i]) \leq \td(G[Y_i]) \leq d$ as	claimed, too. 

	Finally, let $G':=G[Y_0\cup Y_1'\cup\dots\cup Y_\ell']$
	and $\probpar':=\probpar+\mu_1+\dots+\mu_\ell$.
	The equivalence of the instances $(G,\probpar)$ and $(G',\probpar')$
	of $\Pi$ then immediately follows from the safety of the protrusion replacement
	Rule~\ref{rrule:Protrusion}.
\end{proof}

\noindent
With the lemmas at hand we can now prove the main theorem of this section.

\begin{proof}[Proof of Theorem~\ref{theorem:KernelBoundedExp}]
    Given an instance $(G,\probpar)$ of $\Pi$ with $G \in \mathcal K$, 
    we calculate a $2^d$-approximate modulator $S$ using
    Lemma~\ref{lemma:TDModApprox}. Using the algorithm outlined in
    the proof of Lemma~\ref{lemma:Decomposition}, we compute the
    decomposition $\YYYY$. Each cluster $Y_i, 1 \leq i \leq \ell$,
    forms a protrusion with boundary size $|N(Y_i)| \leq 2^{d+1}+2f(2^d) =: h$
    and treedepth (and thus treewidth) at most~$d$. 

    Applying Lemma~\ref{lemma:ProtrusionReplacement} now yields an
    equivalent instance $(G',\probpar')$ with $|V(G')| = |Y_0| + \sum_{i=1}^{\ell} |Y'_i|$
	vertices, where $Y'_i$ denote the clusters obtained through applications of the reduction rule. 
	This quantity is at most $\bigO(|S|) + \ell \cdot \rho(d+h,d+h) = \bigO(|S|)$
	by Lemma~\ref{lemma:Decomposition}\,(3).
	As $G'$ is an induced subgraph of $G$, the above implies that $|V(G')|+|E(G')| = \bigO(|S|)$
	by the degeneracy of $G$, and that $G' \in \mathcal K$.
\end{proof}

\subsection{Problems having finite integer index}
\label{sub:problemsFII}

Several graph problems have finite integer index on the class of all graphs and thus admit linear kernels
on graphs of bounded expansion if parameterized by a treedepth modulator.

\begin{corollary}\label{cor:fiiGeneral}
The following graph problems have finite integer index, and hence have linear kernels
in graphs of bounded expansion, when the parameter is the size of a modulator to constant treedepth:
\name{Dominating Set}, 
\name{$r$-Dominating Set},
\name{Efficient Dominating Set},
\name{Connected Dominating Set},
\name{Vertex Cover},
\name{Hamiltonian Path},
\name{Hamiltonian Cycle},
\name{Connected Vertex Cover},
\name{Independent Set},
\name{Feedback Vertex Set}, 
\name{Edge Dominating Set},
\name{Induced Matching},
\name{Chordal Vertex Deletion},
\name{Odd Cycle Transveral},
\name{Induced $d$-Degree Subgraph},
\name{Min Leaf Spanning Tree},
\name{Max Full Degree Spanning Tree}.
\end{corollary}
\noindent For a more comprehensive list of problems that have \FII in general graphs (and hence fall under 
the purview of the above corollary), see~\cite{BFLPST09}.

Some problems do not have \FII in general (see~\cite{Flu97}) but only when restricted to graphs of bounded
treedepth, and for those we have the same conclusion in the following: 

\begin{lemma}\label{lemma:LongestPathFII}
        Let $\cal D$ be a graph class of bounded treedepth.
        Then the problems \name{Longest Path}, \name{Longest Cycle}, \name{Exact $s,t$-Path},
        \name{Exact Cycle} have \FII in $\cal D$.
\end{lemma}
\begin{proof}
	Let $\Pi$ be any one of the mentioned problems, and let $d,t$ be constants
	such that all graphs in $\cal D$ have treedepth~$\leq d$.
	Consider the class $\G_t$ of all $t$-boundaried graphs,
	and let $T = \{0,1,\dots,t\}$.
	
	We define a \emph{configuration} of $\Pi$ with respect to $\G_t$ as a multiset
	\[C = \{(s_1,d_1,t_1),\dots,(s_p,d_p,t_p)\}\]
	of triples from ${(T \times \N \times T)}$.
	We say a $t$-boundaried graph $\widetilde G \in \G_t$ \emph{satisfies} the configuration
	$C$ if there exists a set of (distinct) paths $P_1,\dots,P_p$ in $G$ such that
	\begin{itemize}
		\item $s_i,t_i$ can only be endvertices of $P_i$,
		$V(P_i) \cap \bound(G) \subseteq \{s_i,t_i\}$, 
		and $|P_i| = d_i$, for $1 \leq i \leq p$,
		\item $V(P_i) \cap V(P_j) \subseteq \bound(G)$ for $1 \leq i<j \leq p$,
		\item $V(P_i) \cap V(P_j) \cap V(P_k) =\emptyset$ for $1 \leq i<j<k \leq p$.
	\end{itemize}	
	Note that, for simplicity, we identify the boundary vertices in
	$\bound(G)$ with their labels $1,\dots,t$ from~$T$.
	Moreover, $s_i,t_i$ can take the value $0$ which is not contained in $\bound(G)$: 
	semantically these tuples describe paths which intersect the boundary of 
	$G$ at only one or no vertex. Another special case are tuples with 
	$s_i = t_i$ and $d = 0$: those describe single vertices of the boundary.
	In short, a graph satisfies a configuration if it contains internally 
	non-intersecting paths of length and endvertices 
	prescribed by the tuples of the configuration, and no three of the
	paths are prescribed to have the same endvertex (hence some
	configurations are not satisfiable at all, but this is a small technicality).
	
	The \emph{signature} $\sigma[\widetilde G]$ of a graph $\widetilde G \in \G_t$ is a function from
	the configurations into $\{0,1\}$ where $\sigma[\widetilde G](C) = 1$ iff $\widetilde G$ satisfies $C$. 
	We define:
	$$
		\widetilde{G}_1 \simeq_\sigma \widetilde{G}_2 ~\Longleftrightarrow~ 
				\sigma[\widetilde{G}_1] \equiv \sigma[\widetilde{G}_2]
		 ~\text{for}~ \widetilde{G}_1,\widetilde{G}_2 \in \G_t.
	$$
	We claim that the equivalence relation $\simeq_\sigma$
	is a refinement of $\equiv_{\Pi,t}$.
	We provide only a sketch for $\Pi=$\,\name{Longest Path}, the proofs for the other problems
	work analogous.	
	To this end we assume the contrary, that $\sigma[\widetilde{G}_1] \equiv
	\sigma[\widetilde{G}_2]$ while $\widetilde{G}_1 \not \equiv_t \widetilde{G}_2$.
	Up to symmetry, this means that for all integers $c$ there exists a graph
	$\widetilde{G}_3 \in\G_t$
	such that $(\widetilde{G}_1 \oplus \widetilde{G}_3,\ell) \in \Pi$ but
	$(\widetilde{G}_2 \oplus \widetilde{G}_3,\ell + c) \not \in \Pi$. 
	We choose $c = 0$ and show the contradiction. 
	Thus the graph $\widetilde{G}_1 \oplus \widetilde{G}_3$ contains a 
	path $P$ of length $\ell$ but $\widetilde{G}_2 \oplus \widetilde{G}_3$ does not. 

	Using the implicit order given through the vertex order of $P$ we sort the subpaths of $P$
	contained in $P \cap G_1$ and so obtain a sequence of paths
	$P_1,\dots,P_q\subseteq G_1$, each with at most two vertices -- the ends, in $\bound(G_1)$.
	By identifying each subpath $P_i$ with the tuple $(s_i,d_i,t_i)$ where $d_i = |P_i|$
	and $s_i$ is the label of the start of $P_i$ in $\bound(G_1)$ 
	(or $0$ if $s_i \not \in \bound(G_1)$)
	and $t_i$ the label of the end of $P_i$ in $\bound(G_1)$ (ditto), we
	obtain a configuration $C_P = \{(s_1,d_1,t_1),\dots,(s_q,d_q,t_q)\}$. 
	Now, $\widetilde{G}_1$ satisfies $C_P$ by the definition.
	Since $\sigma[\widetilde{G}_1](C_P)=\sigma[\widetilde{G}_2](C_P)$, there exists a set of
	paths $Q_1,\dots,Q_q\subseteq G_2$ witnessing that $\widetilde{G}_2$ satisfies $C_P$.
	But then $Q_1,\dots,Q_q$ together with $P\cap G_3$ form a path $Q$ of
	length $\ell$ in $\widetilde{G}_2 \oplus \widetilde{G}_3$, a contradiction.
	
	Second, although $\simeq_\sigma$ is generally of infinite index,
	we claim that for every~$t$, only a finite number of equivalence
	classes of $\simeq_\sigma$ carry a representative
	of treedepth~$\leq d$, and hence $\simeq_\sigma$ is of finite index
	when restricted to graphs from~$\cal D$.
	This is rather easy since graphs of treedepth $\leq d$ do not
	contain paths of length $2^d-1$ or longer, and so a graph
	$\widetilde G\in\cal D_t$ can satisfy a configuration $C =
	\{(s_1,d_1,t_1),\dots,(s_p,d_p,t_p)\}$ only if
	$d_i\in\{0,1,\dots,2^d-2\}$ for $1 \leq i \leq p$.
	Recall, each boundary vertex label occurs at most
	twice among $s_1,t_1,\dots,s_p,t_p$ in a satisfiable configuration.
	Hence only finitely many such configurations $C$ can be
	satisfied by a graph from $\cal D_t$, and consequently,
	finitely many function values of $\sigma[\widetilde G]$ are nonzero for any 
	$\widetilde G\in\cal D_t$ and the number of the nonempty classes of 
	$\simeq_\sigma$ restricted to $\cal D_t$ is finite.
\end{proof}

\noindent
Another exemplary problem which does not have FII on general graphs, but does
so on a restricted graph class, is the \textsc{Branchwidth} problem which is 
defined as follows.

A \emph{branch-decomposition} of a graph $G$ is a pair $(T,\tau)$ where $T$ is a
tree of maximum degree three and $\tau$ a bijective function $\tau:
E(G)\rightarrow \{t: \text{$t$ is a leaf of $T$}\}$.  For an edge $e$ of $T$,
the connected components of $T\setminus e$ induce a bipartition $(X,Y)$ of the
edge set of $G$.  The \emph{width} of $e$ is then defined as the number of
vertices of $G$ incident both with an edge of $X$ and an edge of~$Y$.  The
\emph{width} of $(T,\tau)$ is the maximum width over all edges of $T$.  The
\emph{branchwidth}  of $G$ is the minimum of the width of all
branch-decompositions of $G$.
The branchwidth of a graph class is bounded if and only if its treewidth is
bounded.  The \textsc{Branchwidth} problem is, given a graph~$G$ and an
integer~$k$, to decide whether $G$ has branchwidth at most $k$.  

\begin{lemma}\label{lemma:BranchwidthFII}
        Let $\cal B$ be a graph class of bounded branchwidth.
        Then \name{Branchwidth} has \FII in $\cal B$.
\end{lemma}
\begin{proof}
Let $\G_t$ be the class of all $t$-boundaried graphs.
Let $\cal X^w$ denote the set of minor-minimal graphs of branchwidth
greater than $w$ (we will see that $\cal X^w$ is finite
for every $w$ but that is not important for now). 
That is, $G\in\cal X^w$ if and only if the branchwidth of $G$ is
$>w$ but every proper minor of $G$ has branchwidth $\leq w$;
$\>G$ is an ``obstruction'' to branchwidth~$w$.
Let $\cal X^w_{*t}\subseteq\G_t$ be the ``$t$-boundaried fragments'' of members of 
$\cal X^w$, \ie 
$$\widetilde{F}\in \cal X^w_{*t} ~\iff~
\exists \widetilde{F}': \widetilde{F}\oplus \widetilde{F}'\in \cal X^w .$$

Let $\Pi$ be the problem \name{Branchwidth}.
The framework of the proof is very similar to that of
Lemma~\ref{lemma:LongestPathFII}; members of $\cal X^w_{*t}$ play the role
of configurations of $\Pi$ and a signature is a subset of 
$\cal X_{*t}:=\bigcup_{w}\cal X^w_{*t}$.
First, for a $t$-boundaried graph $\widetilde G$, the {\em signature} $\sigma[\widetilde G]$ is defined
as the set of those $\widetilde F\in\cal X_{*t}$ such that $\widetilde F$ is {\em rooted minor}
of $\widetilde G$, meaning that $F$ is a minor of $G$ in such a way that the
boundary $\bound(F)=\bound(G)$ is identical (not touched).
It is routine to verify that if, informally,
the same fragments of ``branchwidth obstructions'' exist in both
$\widetilde{G}_1$ and $\widetilde{G}_2$, then they are equivalent.
Formally;
$$\mbox{if } \sigma[\widetilde{G}_1]=\sigma[\widetilde{G}_2]
\mbox{, then }
\widetilde{G}_1 \equiv_{\Pi,t} \widetilde{G}_2
\mbox{ with } \Delta_{\Pi,t}(\widetilde{G}_1,\widetilde{G}_2)=0 .$$

Second, the equivalence relation $\simeq_{\sigma}$ on $\G_t$ defined by the same
signature $\sigma$ is generally of infinite index, though, we claim that
for every~$b,t$, only a finite number of equivalence
classes of $\simeq_{\sigma}$ carry a representative of branchwidth~$\leq b$.
This would follow if we proved that only finitely many elements of
$\cal X_{*t}$ have branchwidth~$\leq b$.
The latter is a nontrivial statement, possible thanks to some fine properties of
the ``branchwidth obstructions'' as proved in
\cite{GGRW03} (note that although the paper deals with matroids, its
results apply to graph branchwidth as well thanks to~\cite{HicksM07}).
Precisely, besides finiteness of $\cal X^w$ for each $w$, the following
claim \cite[Lemma~4.1]{GGRW03} is used:
\begin{itemize}\item[]
If $\widetilde{F},\widetilde{F}'$ are $t$-boundaried graphs such that 
$\widetilde{F}\oplus \widetilde{F}'\in \cal X^w$ and $t\leq w$,
then $|E(F)|\leq g(t)$ or $|E(F')|\leq g(t)$, where $g(t)=(6^{t-1}-1)/5$.
\end{itemize}

Assume now $\widetilde{F}\in \cal X_{*t}$ such that $F$ is of branchwidth~$b$,
and let $w_0=b+g(t)$.
Either, $\widetilde{F}\in\bigcup_{w<w_0}\cal X^w_{*t}$ which is a finite set,
or there is $\widetilde{F}'$ such that $\widetilde{F}\oplus \widetilde{F}'\in \cal X^w$ where $w\geq w_0$.
If $|E(F')|\leq g(t)$, then the branchwidth of $F$ is greater than
$w-|E(F')|\geq b+g(t)-g(t)=b$, a contradiction.
Therefore, by \cite{GGRW03}, we have $|E(F)|\leq g(t)$ and there are only
finitely many such $t$-boundaried graphs without isolated vertices.
\end{proof}

\noindent
Somehow surprisingly, it is not at all easy to extend the statement of
Lemma~\ref{lemma:BranchwidthFII} to the related problems \textsc{Pathwidth}
and \textsc{Treewidth}, since we have got no direct analogue of the results of
\cite{GGRW03} for the other measures.
See Section~\ref{sec:fii_pw_tw} for further details.

\begin{corollary}\label{cor:fiiBoundedTreedepth}
	The problems 
		\name{Longest Path}, 
		\name{Longest Cycle},
		\name{Exact $s,t$-Path}, 
		\name{Exact Path}, 
	 	and \name{Branchwidth}
	have  linear kernels in graphs of bounded expansion 
	with the size of a modulator to constant treedepth as the parameter.
\end{corollary}

\subsection{Extension to larger graph classes}

\noindent
We can extend our result to classes of graphs that are \emph{nowhere dense},
which present a wider framework than classes of bounded expansion.

\begin{definition}[Nowhere dense~\cite{NOdM10,NOdM11}]
	A graph class $\cal K$ is \emph{nowhere dense} if for all $r \in \N$ it holds
	that $\omega(\cal K \nab r) < \infty$.
\end{definition}

\noindent
In the above definition we use the natural extension of $\omega$ to classes of graphs via
$\omega(\cal K) = \sup_{G \in \cal K} \omega(G)$. Note that nowhere dense classes are closed under
taking shallow minors in the sense that $\cal K \nab r$ is nowhere dense
if $\cal K$ is, albeit with a different bound on the clique size of $r$-shallow minors.

We claim the following kernelization result for
nowhere dense classes, which in particular applies to all problems listed in Section~\ref{sec:Kernels}.

\begin{theorem}\label{theorem:KernelNowhereDense}
    Let a class $\cal K$ be hereditary and nowhere dense and let~$d \in \N$ be a constant.
	Let~$\Pi \in \PPP$. There exist an algorithm that takes as input
    $(G,\probpar) \in \cal K \times \N$ and, in time $\bigO(|G|^{1+\varepsilon})$ for
    every~$\varepsilon > 0$, it outputs
    $(G',\probpar')$ such that
	\begin{enumerate}
            \item $(G,\probpar) \in \Pi$ if and only if $(G',\probpar') \in \Pi$;
            \item $G'$ is an induced subgraph of $G$; and
            \item  for every  $\varepsilon > 0$, 
		$|G'| = \bigO(|S|^{1+\varepsilon})$, where~$S$ is an optimal treedepth-$d$
                modulator of~$G$.
    \end{enumerate}
\end{theorem}

\noindent 
Here we use the nowhere-dense variant of Lemma~\ref{lemma:TDModApprox} to obtain
an approximate treedepth-modulator in almost linear time.
The proof of~\ref{theorem:KernelNowhereDense} follows analogously to the  proof of
\ref{theorem:KernelBoundedExp}, while replacing
Lemma~\ref{lemma:BipartiteNab1} with Lemma~\ref{lemma:BipartiteGeneral} (see
below) and using the following property of nowhere dense classes:

\begin{proposition}[\cite{NOdM08}, also {\cite[Section 5.4]{NOdM12}}]\label{prop:NowhereDenseEdg}
	Let~$\cal G$ be a nowhere dense graph class.
	Then for every~$\alpha > 0$ and every~$r \in\N$ there exists
	$n_{\alpha,r} \in\N$ such that for every $G \in \cal G$ with
	$|G|>n_{\alpha,r}$ it holds that
			$
				\grad_r(G) \leq |G|^\alpha
			$.
\end{proposition}

\noindent
We need additional notation.
For a graph class $\G$ and an integer $p$ we let 
$\G_{\leq p} := \{ H \in \G \mid |H| \leq p \}$ denote those graphs of 
$\G$ which have at most $p$ vertices.
We shortly write $G_{\leq p}$ for $(G\nab0)_{\leq p}$.

\begin{lemma}\label{lemma:BipartiteGeneral}
	Let $G=(X,Y,E)$ be a bipartite graph, and
	$p\geq\grad_1\big(G_{\leq|X|^2}\big)$. Then there are at most 
	\begin{enumerate}
		\item $2p\cdot|X|$ vertices in $Y$ with degree 
			greater than $\omega(G\nab1)$;
		\item $(2p)^{\omega(G\nab1)}\cdot|X|$ 
			subsets $X' \subseteq X$ such that $X' = N(u)$ for some 
			$u \in Y$.
	\end{enumerate}
\end{lemma}
\begin{proof}
	We construct a sequence of graphs $G_0:=G,G_1,\dots,G_\ell$ in the same
	way as in the proof of Lemma~\ref{lemma:BipartiteNab1}.
	Recall that $G_i \in G \nab 1$ for $1 \leq i \leq \ell$, and so
	$\omega(G_\ell[X])\leq\omega(G\nab1)$, in particular.
	Furthermore, since every step $i$ of the sequence adds an edge to
	$G_i[X]$, we have $\ell<|X|^2/2$ and, consequently, $G_\ell[X]$
	results by contracting at most $|X|^2/2$ vertices from $Y$ and so
	$G_\ell[X]\in G_{\leq|X|^2}\nab1$.
	Then $G_\ell[X]$ is actually $2p$-degenerate and
	the first claim follows in exactly the same way as in~\ref{lemma:BipartiteNab1}.

	For the second claim, consider again the set $Y' = Y \cap V(G_\ell)$. The neighborhood of every 
	vertex $v \in Y'$ induces a clique in $G_\ell[X]$, as in Lemma~\ref{lemma:BipartiteNab1}.
	We additionally need a strengthening of Proposition~\ref{prop:DegenerateNumCliques}:

	Assume a graph $H$ and $v\in V(H)$ of degree $d$.
	Then the number of cliques in $H$ containing $v$ is clearly at most
	$\sum_{i=1}^{\omega(H)-1}{d\choose i}\leq d^{\omega(H)-1}$.
	If $H$ is $d$-degenerate, the overall number of cliques in $H$ is
	thus at most $d^{\omega(H)-1}\cdot|H|$.
	In our case of $H=G_\ell[X]$, there are at most	$(2p)^{\omega(G\nab1)-1}\cdot|X|$
	possible cliques in $G_\ell[X]$.
	This quantity accounts for all possible distinct neighborhoods of
	vertices of $Y'$ in $X$, and summing with at most
	$\ell\leq2p\cdot|X|$ neighborhoods of the
	vertices of $Y \setminus V(G_\ell)$ we get (with a large margin) the bound in the second claim.
\end{proof}

\noindent
The following two corollaries are analogues of
Corollary~\ref{cor:LargeDegree} and~\ref{cor:NumClusters} and will be used
in a similar fashion.

\begin{corollary}\label{cor:LargeDegreeGeneral}
	Let $\cal K$ be a nowhere dense graph class, and fix any $\varepsilon>0$ and $\delta\in\N$. 
	Let $q=\omega(\cal K\nab(\delta+1))<\infty$.
	There exists $n_0\in\N$, depending on $\cal K$ and $\varepsilon,\delta$, 
	such that the following holds for every $G \in \cal K$ and
	$S \subseteq V(G)$, $|S|>n_0$:
	If $C_1, \ldots, C_s$ are disjoint connected subgraphs of $G-S$
	satisfying $\diam(G[V(C_i)]) \leq \delta$ and $|N_S(C_i)| > q$
	for $i=1,\ldots,s$, then $s\leq|S|^{1+\varepsilon}$.
\end{corollary} 
\begin{proof}
	We construct an auxilliary  bipartite graph $\bar G$ with partite 
	sets $S$ and $Y = \{C_1,\ldots,C_s\}$. There is an edge between $C_i$ 
	and $x \in S$ iff $x \in N_S(C_i)$.
	As in Corollary~\ref{cor:LargeDegree}, we know that $\bar G$ is a 
	depth-$\delta$ shallow minor of $G$ with branch sets $C_i, 1 \leq i \leq s$,
	and, for any $F\in\bar G\nab1$, it is moreover $F\in G\nab(\delta+1)$.
	In particular, $\omega(\bar G\nab1)\leq\omega(G\nab(\delta+1))\leq q$.
	Though, we will also need the following small refinement of the previous fact:

	Clearly, there exists a connected subgraph $C_i'\subseteq C_i$ such
	that $N_S(C_i')=N_S(C_i)$, $\diam(C_i')\leq2\delta$ and
	$|C_i'|\leq \diam(G[V(C_i)])\cdot|N_S(C_i)|+1<2\delta|S|$\,---simply
	take a vertex $w\in V(C_i)$ together with shortest paths from $w$ to
	selected neighbors of $N_S(C_i)$ in~$C_i$.
	Hence it holds for any $F\in\bar G_{\leq|S|^2}\nab1$ that
	$F\in G_{\leq m}\nab(2\delta+1)$ where $m=|S|^2\cdot2\delta|S|=2\delta|S|^3$.

	Then, using also Proposition~\ref{prop:NowhereDenseEdg},
	$\grad_1\big(\bar G_{\leq|S|^2}\big)\leq
	 \grad_{2\delta+1}\big(G_{\leq m}\big)\leq m^\alpha$ for any
	$\alpha>0$ and all sufficiently large $|G|$ and~$m$.
	We choose $\alpha=\varepsilon/4$.
	By the first claim of Lemma~\ref{lemma:BipartiteGeneral}, for
	$p=m^\alpha$, we get that
	$$s\leq2p|S|=2\big(2\delta|S|^3\big)^{\varepsilon/4}\cdot|S|<|S|^{1+\varepsilon}$$
	whenever $|S|$ is sufficiently large.
\end{proof}

\begin{corollary}\label{cor:NumClustersGeneral} 
	Let $\cal K$ be a nowhere dense graph class, and fix any $\varepsilon>0$ and $\delta\in\N$. 
	There exists $n_0\in\N$, depending on $\cal K$ and $\varepsilon,\delta$, 
	such that the following holds for every $G \in \cal K$ and
	$S \subseteq V(G)$, $|S|>n_0$:
	If $\mathcal C_1, \ldots, \mathcal C_t$ are sets of
	connected components of $G-S$ such that
  \begin{itemize}
  \item[--]for all $C,C' \in \bigcup_i \mathcal{C}_i$ it holds
	that $C, C' \in \mathcal C_j$ for some $j$ if and only if $N_S(C)=N_S(C')$, and
  \item[--]for all $C \in \bigcup_i \mathcal{C}_i$, $\diam(G[V(C)]) \leq\delta$,
  \end{itemize}
	then $t\leq|S|^{1+\varepsilon}$.
\end{corollary}
\begin{proof}
	Let $q=\omega(\cal K\nab(\delta+1))<\infty$.
	As in the proof of Corollary~\ref{cor:LargeDegreeGeneral}, we 
	construct a bipartite graph $\bar G$ with partite sets $S$ 
	and $Y = \{C_1, \ldots, C_r\}$, where the vertices $C_j$,
	$1\leq j\leq r$, represent the connected components in $\bigcup_i \mathcal{C}_i$ 
	and $C_j$ has an edge to $x \in S$ iff $x \in N_S(C_j)$.
	As before, it holds for any $F\in\bar G_{\leq|S|^2}\nab1$ that
        $F\in G_{\leq m}\nab(2\delta+1)$ where $m=2\delta|S|^3$,
	and consequently $\grad_1\big(\bar G_{\leq|S|^2}\big)\leq
         \grad_{2\delta+1}\big(G_{\leq m}\big)\leq m^\alpha$
	for any $\alpha>0$ and all sufficiently large $|G|$ and~$m$.
	
	We now choose $\alpha=\varepsilon/(4q)$ and apply the second claim
	of Lemma~\ref{lemma:BipartiteGeneral};
	\begin{align*}
		t & \leq |\{S' \subseteq S \mid \exists C_i \in Y : N(C_i) = S' \}| \\ 
		  & \leq (2m^\alpha)^{\omega(\bar G\nab1)}\cdot|S|
			\leq 2^q m^{\alpha q}\cdot|S| \\
		  & = 2^q m^{\varepsilon/4}\cdot|S| =
			2^q(2\delta)^{\varepsilon/4}|S|^{3\varepsilon/4}\cdot|S|
			< |S|^{1+\varepsilon}
	\end{align*}
	whenever $|S|$ is sufficiently large.
\end{proof}

\noindent
We are now ready to prove the theorem.
First, the following generalization of Lemma~\ref{lemma:Decomposition} 
follows easily using the above two corollaries.

\begin{lemma}\label{lemma:DecompositionGeneral}
	Let $\cal K$ be a nowhere dense graph class, and fix any $\varepsilon>0$ 
	and $d\in\N$ ($d$ a constant). 
	Let $q=\omega(\cal K\nab 2^d)<\infty$.
	Assume any $G \in \mathcal{G}$ and $S \subseteq V(G)$ a set of vertices such
	that\/ $\td(G-S) \leq d$.
	There is an algorithm that runs in time linear in $|G|$ and 
	partitions $V(G)$ into sets $\YYYY$ such that the following hold:
	\begin{enumerate}
	\item $S \subseteq Y_0$ and $|Y_0| = \bigO(|S|^{1+\varepsilon})$; 
	\item for $1 \leq i \leq \ell$, $Y_i$ induces a set of 
		connected components of $G-Y_0$ that have 
		the same neighborhood in $Y_0$ of size at most 
		$2^{d+1}+q$;
	\item $\ell \leq \bigO(|S|^{1+\varepsilon})$.
	\end{enumerate}
\end{lemma}
\begin{proof}
	We use the same algorithm as in the proof of Lemma~\ref{lemma:Decomposition};
	setting the size of a large neighborhood to $q+1$
	in accordance with the bound in	Corollary~\ref{cor:LargeDegreeGeneral}.
	This proves the first two claims, provided $|S|$ is sufficiently large.
	The third claim then follows from the conclusion of
	Corollary~\ref{cor:NumClustersGeneral}.
	If, on the other hand, $|S|$ is bounded from above by a constant,
	then the claims follow from any trivial estimates;
	e.g., $s\leq|S|^2$ in place of Corollary~\ref{cor:LargeDegreeGeneral} and
	$t\leq|S|^q$ in place of Corollary~\ref{cor:NumClustersGeneral}.
\end{proof}

\begin{proof}[Proof of Theorem~\ref{theorem:KernelNowhereDense}]
	The proof now proceeds in exactly the same way as that of  Theorem~\ref{theorem:KernelBoundedExp}.
\end{proof}

\section{Structural Parameterizations of \Problem{Treewidth} and \Problem{Pathwidth}}\label{sec:Treewidht}
\newcommand{\FSCP}{\textup{FSCP}\xspace}
\newcommand{\FSCT}{\textup{FSCT}\xspace}
\label{sec:fii_pw_tw}

The purpose of this section is to prove that the problems \textsc{Pathwidth} and
\textsc{Treewidth} have \FII on graphs of bounded pathwidth and treewidth,
respectively. We start by showing that this is not the case on the class of all graphs.
\begin{proposition}\label{prop:nofii}
  The problems \textsc{Pathwidth} and \textsc{Treewidth} do not have
  \FII on the class of all graphs.
\end{proposition}
\begin{proof}
  For $w,t \in \N$, let $\widetilde{G}_w=(G_w,\partial(G_w))$ be the
  $t$-boundaried complete graph with $w+t$ vertices.
  We claim that 
  $\widetilde{G}_{w} \not \equiv_{\pw,t} \widetilde{G}_{w+1}$ and $\widetilde{G}_{w} \not \equiv_{\tw,t}
  \widetilde{G}_{w+1}$ for every $w \in \N$ with $w > t$. This shows that neither
  $\equiv_{\pw,t}$ nor $\equiv_{\tw,t}$ is finite and concludes the
  proof of the theorem.

  Let $\widetilde{H}_1=\widetilde{G}_w$ and $\widetilde{H}_2=\widetilde{G}_{w+1}$. Then,
  $\pw(\widetilde{G}_w \oplus \widetilde{H}_1)=\tw(\widetilde{G}_w
  \oplus \widetilde{H}_1)=t+w$
  and 
  $\pw(\widetilde{G}_{w+1} \oplus \widetilde{H}_1)=\tw(\widetilde{G}_{w+1}
  \oplus \widetilde{H}_1)=t+w+1$
  but
  $\pw(\widetilde{G}_w \oplus \widetilde{H}_2)=\tw(\widetilde{G}_w
  \oplus \widetilde{H}_2)=t+w+1$
  and 
  $\pw(\widetilde{G}_{w+1} \oplus \widetilde{H}_2)=\tw(\widetilde{G}_{w+1}
  \oplus \widetilde{H}_2)=t+w+1$, as required.
\end{proof}

\noindent
In the rest of the section we focus on proving that the problems \textsc{Pathwidth} and \textsc{Treewidth} have \FII on graphs of bounded pathwidth and treewidth,
respectively.
Compared to the path and cycle problems treated in
Lemma~\ref{lemma:LongestPathFII} and the branchwidth problem as in
Lemma~\ref{lemma:BranchwidthFII}, the proofs here are much more involved and
use the notion of characteristics of path
decompositions and tree decompositions, which have been introduced
in~\cite{BK96}. Because
the definition of these characteristics is quite technical and the
properties we require have
already been shown in~\cite{BK96}, we will
not provide a formal definition.
Instead, we will only state the required properties 
and refer the reader to~\cite{BK96} for details and proofs.
 
The concept of a \emph{characteristic} of a partial path decomposition of a graph---or
equivalently the characteristic of a path decomposition of a boundaried graph---was 
introduced by Bodlaender and Kloks in~\cite[Definition~4.4]{BK96}.
Informally, the characteristic of a path decomposition $\mathcal{P}$ of $\widetilde{G}$ 
compactly represents all the information required
to compute, for any $\widetilde{H}$, the ways $\mathcal{P}$ can be extended into a path decomposition of
the graph $\widetilde{G}\oplus \widetilde{H}$.
This information can then
be used to compute the pathwidth of the graph $\widetilde{G}\oplus \widetilde{H}$.
Importantly, the number
of characteristics of path decompositions of width at most $w$ of any
$t$-boundaried graph only depends on $t$ and $w$, but not on the
the graph itself.

\begin{proposition}[{{\rm \cite[Lemma 4.1]{BK96}}}]\label{pro:pw-number-char}
  Let $\widetilde{G}$ be a $t$-boundaried graph and $w$ an integer. Then
  the number of characteristics of path decompositions of width at
  most $w$ of $\widetilde{G}$ is bounded by a function of $t$ and $w$.
\end{proposition}

\noindent
For integer $w$, the \emph{full set of (path decomposition)
  characteristics} of $\widetilde{G}$ of width at most
$w$ (as defined in~\cite[Definition 4.6]{BK96}), denoted
by $\FSCP_w(\widetilde{G})$, is the set of all characteristics of path
decompositions of $\widetilde{G}$ of width at most $w$. We denote by
$\FSCP(\widetilde{G})$ the (possibly infinite) set $\bigcup_{w \in
  \N}\FSCP_w(\widetilde{G})$.
  
\begin{proposition}[{\cite[Section 4.3]{BK96}}]\label{pro:pw-characteristics}
  Let $\widetilde{H}$, $\widetilde{G}_1$ and
  $\widetilde{G}_2$ be $t$-boundaried graphs, and let $\mathcal{P}$ be a path
  decomposition of $\widetilde{G}_1 \oplus \widetilde{H}$. If the
  (unique) characteristic of $\mathcal{P}|G_1$
  is in $\FSCP(\widetilde{G}_2)$, then there
  is a path decomposition of $\widetilde{G}_2 \oplus \widetilde{H}$
  that has the same width as $\mathcal{P}$. 
\end{proposition}
\begin{proof}[Sketch]
  For $i \in \{1,2\}$, let $\mathcal{P}_i$ be any path decomposition of
  $\widetilde{G}_i$ such that the content of the last bag of $\mathcal{P}_i$
  is $\partial(G_i)$ and let $\mathcal{P}_3$ be any path decomposition of
  $\widetilde{H}$ such that the content of the first bag of $\mathcal{P}_3$ is
  $\partial(H)$. Furthermore, for $i \in \{1,2\}$, let $\mathcal{P}_{i,3}$ be the path
  decomposition of $\widetilde{G}_i \oplus \widetilde{H}$ obtained from $\mathcal{P}_{i}$ and $\mathcal{P}_3$ by appending 
  the first bag of $\mathcal{P}_3$ to the last bag of $\mathcal{P}_i$, let
  $p_{i,3}$ be the bag of $\mathcal{P}_{i,3}$ that corresponds to the last
  bag of $\mathcal{P}_i$, and let $l_{i,3}$ be the last bag of $\mathcal{P}_{i,3}$.

  Now assume that we run the algorithm described in~\cite[Section
  4.3]{BK96} on the path decomposition $\mathcal{P}_{i,3}$ and let
  $F(p_{i,3})$ and $F(l_{i,3})$
  be the full set of characteristics of partial path decompositions
  computed at the node $p_{i,3}$ and the node $l_{i,3}$, respectively, 
  of width at most the width of $\mathcal{P}$. Then, by the definition of a
  full set of characteristics, we obtain that $F(p_{1,3})$ contains
  the characteristic of $\mathcal{P}|G_1$ and that $F(l_{1,3})$ contains
  the characteristic of $\mathcal{P}$. Moreover, the characteristic of $\mathcal{P}$
  in $F(l_{1,3})$ is generated by the algorithm from the
  characteristic of $\mathcal{P}|G_1$ in $F(p_{1,3})$.
  By the assumptions of the
  Proposition, we have that the characteristic of $\mathcal{P}|G_1$ is
  contained in $\FSCP(\widetilde{G}_2)$ and hence also in
  $F(p_{2,3})$. Hence, because the path decompositions $\mathcal{P}_{1,3}$ 
  and $\mathcal{P}_{2,3}$ are identical with respect to everything behind the
  nodes $p_{1,3}$ and $p_{2,3}$, respectively, we obtain that the
  characteristic of $\mathcal{P}$ is also contained in $F(l_{2,3})$,
  witnessing that $\widetilde{G}_2 \oplus \widetilde{H}$ has a path
  decomposition with the same width as $\mathcal{P}$.
\end{proof}

\noindent
The above Proposition illuminates the usefulness of characteristics to show
\FII for the \textsc{Pathwidth} problem. In particular, it follows that if
$\FSCP(\widetilde{G}_1)=\FSCP(\widetilde{G}_2)$, then $\widetilde{G}_1
\equiv_{\pw,t} \widetilde{G}_2$, for all $t$-boundaried graphs
$\widetilde{G}_1$ and $\widetilde{G}_2$. Hence, the full set of
characteristics of a boundaried graph fully describes its equivalence
class with respect to $\equiv_{\pw,t}$. However, as mentioned above
the full set of characteristics of a boundaried graph can be
infinite. We will show in the next section that if we consider \FII
with respect to a class $\cal C$ of graphs of bounded pathwidth, then it is
sufficient to consider the set $\FSCP_{(\pw(\widetilde{G})+t)}(\widetilde{G})$
instead of $\FSCP(\widetilde{G})$ for every $t$-boundary graph $\widetilde{G}=(G,\partial(G))$ with
$G \in \cal C$. Because $\pw(\widetilde{G})$ is bounded by a constant,
the set of characteristics $\FSCP_{(\pw(\widetilde{G})+t)}$ is finite due to
Proposition~\ref{pro:pw-number-char}.

In the following we introduce characteristics for tree decompositions
of boundaried graphs. All the
explanations for characteristics of path decompositions transfer to
characteristics of tree decompositions and we will not repeat them here.
In~\cite[Definition 5.9]{BK96} the authors define the
\emph{characteristic} of a tree decomposition of a boundaried
graph. They show the following:
\begin{proposition}[{\cite[Remark below Lemma 5.3]{BK96}}]
  Let $\widetilde{G}$ be a $t$-boundaried graph and $w$ an integer. Then
  the number of characteristics of tree decompositions of width at
  most $w$ of $\widetilde{G}$ is bounded by a function of $t$ and $w$.
\end{proposition}

\noindent
For an integer $w$, the
\emph{full set of (tree decomposition) characteristics} of $\widetilde{G}$ of width at most
$w$ (as defined in~\cite[Definition 5.11]{BK96}), denoted
by $\FSCT_w(\widetilde{G})$, is the set of all characteristics of tree
decompositions of $\widetilde{G}$ of width at most $w$. We denote by
$\FSCT(\widetilde{G})$ the (possible infinite) set $\bigcup_{w \in
  \N}\FSCT_w(\widetilde{G})$.
\begin{proposition}[{\cite[Section 5.3]{BK96}}]\label{pro:tw-characteristics}
  Let $\widetilde{H}$, $\widetilde{G}_1$ and
  $\widetilde{G}_2$ be $t$-boundaried graphs, and let $\mathcal{T}$ be a tree
  decomposition of $\widetilde{G}_1 \oplus \widetilde{H}$. If the
  (unique) characteristic of $\mathcal{P}|G_1$
  is in $\FSCT(\widetilde{G}_2)$, then there
  is a tree decomposition of $\widetilde{G}_2 \oplus \widetilde{H}$
  that has the same width as $\mathcal{T}$. 
\end{proposition}

\subsection{\textsc{Pathwidth} has FII on graphs of small pathwidth}\label{sec:pwfii}

As stated in the previous section, we will make use of characteristics of
path decompositions of boundaried graphs to show
\FII{} for the \textsc{Pathwidth} problem in a class of graphs of bounded
pathwidth. In particular, we will show that the equivalence relation $\simeq_{\pw,t}$
defined by 
$$\widetilde{G}_1\simeq_{\pw,t} \widetilde{G}_2
	\mbox{~~if and only if~~}
\FSCP_{(\pw(G_1)+t)}(\widetilde{G}_1)=\FSCP_{(\pw(G_2)+t)}(\widetilde{G}_2)$$
is a refinement of the equivalence relation $\equiv_{\pw,t}$.
The following lemma, which we believe to be
interesting in its own right, is central to our proof.
\begin{lemma}\label{lem:pw-projection}
  Let $\widetilde{G}_1,\widetilde{G}_2$ be two $t$-boundaried graphs,
  $G = \widetilde{G}_1 \oplus \widetilde{G}_2$, and
  $\mathcal{P} = (P,\chi)$ be a path decomposition of $G$. Then there is a path
  decomposition $\mathcal{P}' = (P',\chi')$ of $G$ of the same width as $\mathcal{P}$ 
  such that $\mathcal{P}'|G_1$ has width at most \mbox{$\pw(G_1)+t$}.
\end{lemma}

\begin{proof}
  
  If $\mathcal{P}|G_1$ has width at most $\pw(G_1)+ t$, then
  $\mathcal{P}':=\mathcal{P}$ is the required path decomposition of $G$. 
  Otherwise,
  there is a bag $p \in V(P)$ such that $|\chi(p) \cap
  V(G_1)|>\pw(G_1)+t+1$. Call such a bag $p$ a \emph{bad bag} of
  $\mathcal{P}$. The next claim shows that we can eliminate the
  bad bags of $\mathcal{P}$ one by one without introducing new bad bags. 
  Hence, we obtain the desired path decomposition $\mathcal{P}'$ from $\mathcal{P}$
  by a repeated application of the following claim:
  \begin{claim}
    There is a path decomposition $\mathcal{P}''=(P'',\chi'')$ of $G$ of the
    same width as $\mathcal{P}$ such that the set of bad bags of $\mathcal{P}''$ is
    a proper subset of the set of bad bags of $\mathcal{P}$. Moreover, the
    bag $p$ is no longer a bad bag of $\mathcal{P}''$.
  \end{claim}
  \noindent
  Let $\chi_{G_1}(p)$ be the set of vertices $\chi(p) \cap V(G_1)$
  and let $S$ be a minimal separator between $\chi_{G_1}(p)$ and $\partial(G_1)$
  in the graph $G$. Since $\partial(G_1)$ separates $\chi_{G_1}(p)$ from $\partial(G_1)$ 
  and is of cardinality at most $t$, we obtain that $|S| \leq t$.
  Let $W$ be the set of all vertices reachable from $\chi_{G_1}(p)$ in
  $G\setminus S$, and let $\mathcal{P}_{W}=(P_{W},\chi_{W})$
  be an optimal path
  decomposition of $G[W]$. Then, because $W \subseteq V(G_1)$, it follows
  that the width $\mathcal{P}_{W}$ is at most the pathwidth of $G_1$.

  To obtain the desired path decomposition $\mathcal{P}''$, where $p$ is not a
  bad bag anymore, we delete all vertices of $W$
  from the bags of $\mathcal{P}$ and, instead, insert the path decomposition $\mathcal{P}_W$
  between $p$ and an arbitrary neighbor of $p$ in $P$. To ensure Property P3 of
  a path decomposition for the vertices in $\chi(p) \setminus V(G_1)$, we add
  $\chi(p) \setminus V(G_1)$ to every bag of $\mathcal{P}_W$ in $\mathcal{P}''$.
  Furthermore, to cover the
  edges between $S$ and $W$ in $G$ we also need to add $S$ to $p$ and every
  bag of $\mathcal{P}_W$. Because $\chi(p)$ does not necessarily
  contain all vertices of $S$, this could potentially
  violate the Property P3 of a path decomposition. To get around this
  we will add a vertex $s \in S$ to every bag $p' \in V(P)$ in between
  $p$ and any bag containing $s$, i.e.,
  we complete $\mathcal{P}''$ into a valid path decomposition in a minimal way. 
  This completes the construction of $\mathcal{P}''$ and it remains to argue that adding
  these vertices from $S$ does not increase the width of any bag in $\mathcal{P}$.
  Suppose it does, and let $p_2$ be a bag
  where we add more vertices than we remove.
  It follows that there is a bag $p_1 \in V(P)$ such that $p_2$
  lies on the path from $p_1$ to $p$ in $P$ and $|R|<|S'|$, where 
  $R=\chi(p_2) \cap W$ and $S'=(\chi(p_1)\setminus \chi(p_2)) \cap
  S$. 
  Note that in $\mathcal{P}|G[W \cup S']$ we have $\chi_{G_1[W \cup S']}(p_2)= R$.
  Because of Proposition 1 applied to $\mathcal{P}|G_1[W \cup S']$, $	R$ separates $\chi_{G_1[W\cup S']}(p)$ from $S'$ in $G_1[W \cup S']$.  

  We claim that $S''=(S \setminus S')\cup R$ is a separator between
  $\chi_{G_1}(p)$ and $\partial(G_1)$. Since $|S''|<|S|$, this would
  contradict the minimality of $S$.  Let $\Pi$ be a path between
  $\chi_{G_1}(p)$ and $\partial (G_1)$. Since $\chi_{G_1}(p) \subseteq W
  \cup S$, $\Pi$ has to intersect $S$ in order to reach $\partial
  (G_1)$. Let $s$ be the first vertex of $\Pi$ which intersects $S$
  (note that the subpath from $\chi_{G_1}(p$) to $s$ of $\Pi$ lies
  entirely in $W$). Either $s \in S \setminus S'$ and therefore $s \in
  S''$, or $s \in S'$ and the subpath from $\chi_{G_1}(p$) to $s$ of
  $\Pi$ lies entirely in $W \cup S'$, and therefore $\Pi$ has to
  intersect $R \subseteq S''$ in order to reach $s$. It follows that
  $S''$ is indeed a separator between $\chi_{G_1}(p)$ and $\partial
  (G_1)$, completing the proof.
\end{proof}

\noindent
We note here that the bound for the pathwidth given in the above
lemma is essentially tight. To see this consider the complete bipartite
graph $G$ that has $t$ vertices on one side (side $A$) and $t+1$ vertices
on the other side (side $B$). Let $\widetilde{G}_1$ be the graph
$G[A]$ with boundary $A$, let $\widetilde{G}_2$ be the graph $G$
with boundary $A$, and let $\mathcal{P}$ be any optimal path decomposition
of $\widetilde{G}_1\oplus \widetilde{G}_2=G$. Then, because $G$ is a
complete bipartite graph, whose smaller side is $A$, it holds
that $\mathcal{P}$ contains a bag containing $A$. Consequently,
$\pw(\mathcal{P}|G_1)=t-1$ while $\pw(G_1)=0$.

\begin{corollary}\label{cor:pw-projection}
  Let $\widetilde{G}_1$ and $\widetilde{G}_2$ be two $t$-boundaried
  graphs and
  $G=\widetilde{G}_1 \oplus \widetilde{G}_2$. Then there is an optimal path
  decomposition $\mathcal{P}$ of $G$ such that $\mathcal{P}|G_1$ has width at most
  $\pw(G_1)+t$.
\end{corollary}

\noindent
The following lemma shows that $\simeq_{\pw,t}$ is a refinement of $\equiv_{\pw,t}$.
\begin{lemma}\label{lem:pw-refinement}
  Let  $\widetilde{G}_1$ and
  $\widetilde{G}_2$ be two $t$-boundaried graphs. If 
  $\widetilde{G}_1 \simeq_{\pw,t} \widetilde{G}_2$
  , then
  $\widetilde{G}_1 \equiv_{\pw, t} \widetilde{G}_2$.
\end{lemma}
\begin{proof}\sloppypar
  Let $\widetilde{G}_1$ and
  $\widetilde{G}_2$ be two $t$-boundaried graphs such that
  ${\widetilde{G}_1 \simeq_{\pw,t} \widetilde{G}_2}$ and hence
  ${\FSCP_{(\pw(G_1)+t)}(\widetilde{G}_1)=\FSCP_{(\pw(G_2)+t)}(\widetilde{G}_2)}$.
  We show that
  ${\pw(\widetilde{G}_1 \oplus \widetilde{H}) \leq \probpar}$ if and
  only if $\pw(\widetilde{G}_2 \oplus \widetilde{H}) \leq \probpar$
  for any $t$-boundaried graph $\widetilde{H}$ and any $\probpar \in
  \N$. This implies $\widetilde{G}_1 \equiv_{\pw, t} \widetilde{G}_2$
  with $\Delta_{\pw, t}(\widetilde{G}_1,\widetilde{G}_2)=0$.

  Let $\widetilde{H}$ and $\probpar$ be such that
  $\pw(\widetilde{G}_1 \oplus \widetilde{H}) \leq \probpar$. 
  It follows from Corollary~\ref{cor:pw-projection} that there is a path
  decomposition $\mathcal{P}=(P,\chi)$ of $\widetilde{G}_1 \oplus \widetilde{H}$ of width at
  most $\probpar$ such that $\mathcal{P}|G_1$ is a path decomposition of
  $G_1$ of width at most $\pw(G_1)+t$. Hence, there is a
  characteristic in $\FSCP_{(\pw(G_1)+t)}(\widetilde{G}_1)$
  corresponding to $\mathcal{P}|G_1$. Since
  $\FSCP_{(\pw(G_1)+t)}(\widetilde{G}_1)=\FSCP_{(\pw(G_2)+t)}(\widetilde{G}_2)$,
  we have that $\widetilde{G}_2$ has the same characteristic.
  It now follows
  from Proposition~\ref{pro:pw-characteristics} that there is a path
  decomposition of $\widetilde{G}_2\oplus \widetilde{H}$ that has the
  same width as $\mathcal{P}$ and hence $\pw(\widetilde{G}_2 \oplus
  \widetilde{H}) \leq \probpar$, as required. Because the reverse
  direction is analogous, this concludes the proof of the lemma.
\end{proof}

\noindent
We are now ready to show the main result of this subsection, i.e., that the \textsc{Pathwidth} problem has \FII{} on graphs of
bounded pathwidth.
\begin{theorem}\label{the:pw-fii}
  For $w \in \N$, let $\mathcal{PW}_w$ be a class of graphs
  that have pathwidth at most $w$.
  Then, the problem \textsc{Pathwidth}
  has \FII in $\mathcal{PW}_w$.
\end{theorem}
\begin{proof}
  Because $\pw(G_1)\leq w$ and
  $\pw(G_2)\leq w$, it follows from
  Proposition~\ref{pro:pw-number-char}
  that the number of equivalence classes of $\simeq_{\pw,t}$ is
  finite for every $t \in \N$. Furthermore, because of Lemma~\ref{lem:pw-refinement} it
  holds that $\simeq_{\pw,t}$ is a refinement of $\equiv_{\pw, t}$,
  which concludes the proof of the theorem.
\end{proof}

\subsection{\textsc{Treewidth} has FII on graphs of small treewidth}\label{sec:twfii}

As the main ideas of the proof for treewidth are the same as for
pathwidth (see the previous section),
we present in details only the first step, Lemma~\ref{lem:tw-projection},
which is different from former Lemma~\ref{lem:pw-projection}.
\begin{lemma}\label{lem:tw-projection}
  Let $\widetilde{G}_1$ and $\widetilde{G}_2$ be two $t$-boundaried graphs,
  $G=\widetilde{G}_1 \oplus \widetilde{G}_2$, and
  $\mathcal{T}=(T,\chi)$ be a tree decomposition of $G$. Then there is a tree
  decomposition $\mathcal{T}'=(T',\chi')$ of $G$ with the same
  width as $\mathcal{T}$ such that $\mathcal{T}'|G_1$ has
  width at most $\tw(G_1)+t$.
\end{lemma}
\begin{proof}
  If $\mathcal{T}|G_1$ has width at most $\tw(G_1)+t$, then
  $\mathcal{T}':=\mathcal{T}$ is the required tree decomposition of $G$. Hence,
  there is a bag $p \in V(T)$ such that $|\chi(p) \cap
  V(G_1)|>\tw(G_1)+t+1$. We call such a bag $p$ a \emph{bad} bag of
  $\mathcal{T}$. 
  The next claim shows that we can eliminate the
  bad bags of $\mathcal{T}$ one by one without introducing new bad bags. 
  Hence, we obtain the desired tree decomposition $\mathcal{T}'$ from $\mathcal{T}$
  by a repeated application of the following claim.
  \begin{claim}
    There is a tree decomposition $\mathcal{T}''=(T'',\chi'')$ of
    $G$ of the same width as $\mathcal{T}$ such that the set of bad bags of $\mathcal{T}''$ is
    a proper subset of the set of bad bags of $\mathcal{T}$. Moreover, the
    bag $p$ is no longer a bad bag of $\mathcal{T}''$.
  \end{claim}
  \noindent
  Let $\chi_{G_1}(p)$ be the set of vertices in $\chi(p) \cap V(G_1)$
  and let $S$ be a minimal separator between $\chi_{G_1}(p)$ and $\partial(G_1)$
  in the graph $G$. Then, because $\partial(G_1)$ is
  a separator between $\chi_{G_1}(p)$ and $\partial(G_1)$ of
  cardinality at most $t$, we obtain that $|S| \leq t$.
  Let $W$ be the set of all vertices reachable from $\chi_{G_1}(p)$ in
  $G\setminus S$, and let $\mathcal{T}_{W}=(T_{W},\chi_{W})$
  be an optimal tree
  decomposition of $G[W]$. Then, because $W \subseteq V(G_1)$, it follows
  that the width $\mathcal{T}_{W}$ is at most the treewidth of $G_1$.
  
  To obtain the desired tree decomposition $\mathcal{T}''$, where $p$ is not a
  bad bag anymore, we delete all vertices of $W$
  from the bags of $\mathcal{T}$ and, instead, insert the tree decomposition $\mathcal{T}_W$
  by connecting any node of $\mathcal{T}_W$ via an edge to $p$ in $T$. 
  However, to cover the
  edges between $S$ and $W$ in $G$ we also need to add $S$ to $p$ and every
  bag of $\mathcal{T}_W$. Because $\chi(p)$ does not necessarily
  contain all vertices of $S$, this could potentially
  violate the property P3 of a tree decomposition. To get around this
  we will add a vertex $s \in S$ to every bag $p' \in V(T)$ that is on a path
  between $p$ and any bag containing $s$ in $T$, i.e.,
  we complete $\mathcal{T}''$ into a valid tree decomposition in a minimal way. 
  This completes the construction of $\mathcal{T}''$ and it remains to argue that adding
  these vertices from $S$ does not increase the width of any bag in $\mathcal{T}$.
  Suppose it does, and let $p_2$ be a bag
  where we add more vertices than we remove. Let $S' \subseteq S$ be the set of added vertices and $R=\chi(p_2) \cap W$ the set of removed vertices. It follows that $|R|<|S'|$ and the bag $p_2$ separates in $T$ the set of bags containing a vertex from $S'$ from the bag $p$.
  Note that in $\mathcal{T}|G[W \cup S']$ we have $\chi_{G_1[W \cup S']}(p_2)= R$.
  Because of Proposition 1 applied to $\mathcal{T}|G_1[W \cup S']$, $R$ separates $\chi_{G_1[W\cup S']}(p)$ from $S'$ in $G_1[W \cup S']$.

  We claim that $S''=(S \setminus S')\cup R$ is a separator between
  $\chi_{G_1}(p)$ and $\partial(G_1)$. Since $|S''|<|S|$, this would
  contradict the minimality of $S$.  Let $\Pi$ be a path between
  $\chi_{G_1}(p)$ and $\partial (G_1)$. Since $\chi_{G_1}(p) \subseteq W
  \cup S$, $\Pi$ has to intersect $S$ in order to reach $\partial
  (G_1)$. Let $s$ be the first vertex of $\Pi$ which intersects $S$
  (note that the subpath from $\chi_{G_1}(p$) to $s$ of $\Pi$ lies
  entirely in $W$). Either $s \in S \setminus S'$ and therefore $s \in
  S''$, or $s \in S'$ and the subpath from $\chi_{G_1}(p$) to $s$ of
  $\Pi$ lies entirely in $W \cup S'$, and therefore $\Pi$ has to
  intersect $R \subseteq S''$ in order to reach $s$. It follows that
  $S''$ is indeed a separator between $\chi_{G_1}(p)$ and $\partial
  (G_1)$, completing the proof.

\end{proof}

\begin{corollary}\label{cor:tw-projection}
  Let $\widetilde{G}_1$ and $\widetilde{G}_2$ be two $t$-boundaried
  graphs and 
  $G=\widetilde{G}_1 \oplus \widetilde{G}_2$. Then there is an optimal tree
  decomposition $\mathcal{T}$ of $G$ such that $\mathcal{T}|G_1$ has width at most
  $\tw(G_1)+t$.
\end{corollary}

\noindent
Employing a technical lemma analogous to Lemma~\ref{lem:pw-refinement},
we obtain our main result of the subsection.
\begin{theorem}\label{the:tw-fii}
  For $w \in \N$, let $\mathcal{TW}_w$ be a class of graphs
  that have treewidth at most $w$.
  Then, the problem \textsc{Treewidth}
  has \FII in $\mathcal{TW}_w$.
\end{theorem}
\begin{proof}
    The proof is analogous to the proof of Theorem~\ref{the:pw-fii}.
\end{proof}

\noindent
Overall, we can conclude the whole section analogously to
Corollary~\ref{cor:fiiBoundedTreedepth}:

\begin{corollary}\label{cor:fiiSpecialTWPW}
	The problems 
		\name{Pathwidth} and 
		\name{Treewidth}
	have  linear kernels in graphs of bounded expansion 
	with the size of a modulator to constant treedepth as the parameter.
\end{corollary}


\section{Conclusions and Further Research}\label{sec:Conclusion}
We have presented kernelization meta-results on graph classes of
bounded expansion and on nowhere dense classes.  More specifically, we
have shown that all problems with \FII
on graphs of bounded treedepth admit linear problem kernels
on graph classes of bounded expansion when parameterized by the size
of a modulator to constant treedepth. For nowhere dense classes,
we have shown that the kernels have almost-linear size.

The choice of our parameter (treedepth-modulator) is not arbitrary;
as discussed in the introduction, e.g., a modulator to
constant treewidth cannot yield linear kernels for certain natural problems
that one would like to include in the framework.
As argued before, this problem can be resolved {\em only} by
choosing a parameter that generally increases when subdiving edges.
Treedepth, which can be asymptotically characterized by absence of long
paths as a subgraph, is thus a very natural choice for our purpose.

It remains an open question whether polynomial kernels
(under a suitable weaker parametrization) exist for
problems which are not invariant under edge subdivisions, such as
\name{Hamiltonian Cycle}. Furthermore,
our framework is general enough that it might apply to graph classes
which are not part of the sparse graph hierarchy. A meta-kernel result
for a dense graph class would be especially interesting. Recent work
has shown that a linear kernel for classes of bounded expansion and an
almost linear kernel for nowhere dense graph classes for
\name{Dominating Set} exist when parameterized by the natural
parameter \cite{DDFKLPPRSSS14}. This provides some hope that further
problems admit such kernels since \name{Dominating Set} has acted
as a catalyst for a flurry of results before (in fact, it was the problem
that initiated the search for linear kernels on planar graphs).

Finally, it would be interesting to obtain a  natural characterization of
problems that have FII on graphs of bounded treedepth.


\def\redefineme{
    \def\LNCS{LNCS}%
    \def\ICALP##1{Proc. of ##1 ICALP}%
    \def\FOCS##1{Proc. of ##1 FOCS}%
    \def\COCOON##1{Proc. of ##1 COCOON}%
    \def\SODA##1{Proc. of ##1 SODA}%
    \def\SWAT##1{Proc. of ##1 SWAT}%
    \def\IWPEC##1{Proc. of ##1 IWPEC}%
    \def\IWOCA##1{Proc. of ##1 IWOCA}%
    \def\ISAAC##1{Proc. of ##1 ISAAC}%
    \def\STACS##1{Proc. of ##1 STACS}%
    \def\IWOCA##1{Proc. of ##1 IWOCA}%
    \def\ESA##1{Proc. of ##1 ESA}%
    \def\WG##1{Proc. of ##1 WG}%
    \def\LIPIcs##1{LIPIcs}%
    \def\LIPIcs{LIPIcs}%
    \def\LICS##1{Proc. of ##1 LICS}%
}

\bibliographystyle{elsarticle-num}
\bibliography{biblio,conf}

\clearpage

\section{Appendix}\label{sec:Appendix}
In this appendix, we define some of the problems that we mention 
in this paper.  

\begin{defproblem}{Longest Path}
	Input: & A graph $G$ and a positive integer $\ell$.  \\
	Problem: & Does $G$ contain a simple path of length at least $\ell$? 
\end{defproblem}

\begin{defproblem}{Longest Cycle}
	Input:   & A graph $G$ and a positive integer $\ell$.  \\
	Problem: & Does $G$ contain a simple cycle of length at least $\ell$? 
\end{defproblem}

\begin{defproblem}{Exact $s,t$-Path}
	Input: & A graph $G$, two special vertices~$s,t \in V(G)$ and a positive integer $\ell$.  \\
	Problem: & Is there a simple path in $G$ from $s$ to $t$ of length exactly $\ell$? 
\end{defproblem}

\begin{defproblem}{Exact Cycle}
	Input: & A graph $G$ and a positive integer $\ell$.  \\
	Problem: & Is there a simple cycle in $G$ of length exactly $\ell$? 
\end{defproblem}

\begin{defproblem}{Feedback Vertex Set}
	Input:   & A graph $G$ and a positive integer $\ell$.  \\
	Problem: & Is there a vertex set~$S \subseteq V(G)$ with at most $\ell$ vertices such that $G-S$ 
		 is a forest? 
\end{defproblem}

\begin{defproblem}{Treewidth}
	Input:   & A graph $G$ and a positive integer $\ell$.  \\
	Problem: & Is the treewidth of~$G$ at most $\ell$? 
\end{defproblem}

\begin{defproblem}{Pathwidth}
	Input:   & A graph $G$ and a positive integer $\ell$.  \\
	Problem: & Is the pathwidth of~$G$ at most $\ell$? 
\end{defproblem}

\begin{defproblem}{Treewidth-$t$ Vertex Deletion}
	Input:   & A graph $G$ and a positive integer $\ell$.  \\
	Problem: & Is there a vertex set~$S \subseteq V(G)$ with at most $\ell$ vertices
		 such that the treewidth of $G-S$ is at most~$t$? 
\end{defproblem}

\begin{defproblem}{Dominating Set}
	Input:   & A graph $G = (V,E)$ and a positive integer $\ell$.  \\
	Problem: & Is there a vertex set~$S \subseteq V$ with at most $\ell$ vertices
		 such that for all $u \in V \setminus S$ there exists $v \in S$ such that $uv \in E$? 
\end{defproblem}

\smallskip
\noindent If in addition, we require that $G[S]$ is a connected graph then the
problem is called \name{Connected Dominating Set}.

\begin{defproblem}{$r$-Dominating Set}
	Input:   & A graph $G = (V,E)$ and a positive integer $\ell$.  \\
	Problem: & Is there a vertex set~$S \subseteq V$ with at most $\ell$ vertices
		 such that for all $u \in V \setminus S$ there exists $v \in S$ such that $d(u,v) \leq r$? 
\end{defproblem}

\begin{defproblem}{Efficient Dominating Set}
	Input:   & A graph $G = (V,E)$ and a positive integer $\ell$.  \\
	Problem: & Is there an independent set~$S \subseteq V$ with at most $\ell$ vertices
		 such that for every $u \in V \setminus S$ there exists exactly one $v \in S$ such that $uv \in E$? 
\end{defproblem}

\begin{defproblem}{Edge Dominating Set}
	Input:   & A graph $G = (V,E)$ and a positive integer $\ell$.  \\
	Problem: & Is there an edge set~$S \subseteq E$ of size at most $\ell$ 
		 such that for every $e \in E \setminus S$ there exists $e' \in S$ such that $e$ and $e'$ share an endpoint? 
\end{defproblem}

\begin{defproblem}{Induced Matching}
	Input:   & A graph $G = (V,E)$ and a positive integer $\ell$.  \\
	Problem: & Is there an edge set~$S \subseteq E$ of size at least $\ell$ 
		 such that $S$ is a matching and for all $u,v \in V(S)$, if $uv \in E$ then $uv \in S$? 
\end{defproblem}

\begin{defproblem}{Chordal Vertex Deletion}
	Input:   & A graph $G = (V,E)$ and a positive integer $\ell$.  \\
	Problem: & Is there a vertex set~$S \subseteq V$ of size at most $\ell$ 
		 such that $G-S$ is chordal? 
\end{defproblem}

\begin{defproblem}{$\cal F$-Minor-Free Deletion}
	Input:   & A graph $G = (V,E)$ and a positive integer $\ell$.  \\
	Problem: & Is there a vertex set~$S \subseteq V$ of size at most $\ell$ 
		 such that $G-S$ does not contain any graph of the (finite)
		 family~$\cal F$ as a minor? 
\end{defproblem}

\end{document}